\documentclass[11pt,letterpaper]{article} %
\usepackage[T1]{fontenc} \usepackage{circuitikz} %
\usetikzlibrary{positioning, arrows.meta} \usepackage{amsmath, amsthm,
  amssymb, amsfonts, dsfont, mathrsfs, mathtools} %
\usepackage[pdftex,left=1in,top=1in,bottom=1in,right=1in]{geometry} %
\usepackage{graphicx} %
\usepackage{color} %
\usepackage[inline]{enumitem} %
\usepackage{etoolbox} %
\usepackage{todonotes} %
\usepackage[hidelinks,linktocpage=true,pagebackref=true, colorlinks,
linkcolor=black,citecolor=black,urlcolor=black,bookmarks,
bookmarksopen, bookmarksnumbered]{hyperref} %
\usepackage[capitalize,nameinlink,noabbrev]{cleveref} 

\newcommand{\mnoteOK}[1]{} %
\newcommand{\nnoteOK}[1]{} %
\newcommand{\ynoteOK}[1]{} %
\newtoggle{full} 

\toggletrue{full} %
\usepackage{titlesec} %
\usepackage[initials,alphabetic]{amsrefs} %

\title{Optimal Erasure Codes and Codes on Graphs} %

\newcommand{\Omit}[1]{} \renewcommand{\Omit}[1]{#1}

\author{} \Omit{ \author{{\sc Yeyuan Chen, Mahdi Cheraghchi, Nikhil
      Shagrithaya}\thanks{%
      Emails: \{yeyuanch, mahdich, nshagri\}@umich.edu. }\\
    Department of EECS \\
    University of Michigan, Ann Arbor MI, USA }} %
\date{} %

\newcommand{\cC}{\mathcal{C}} %
\newcommand{\U}{\mathcal{U}} %
\newcommand{\F}{\mathds{F}} %
\newcommand{\R}{\mathbb{R}} %
\newcommand{\E}{\mathds{E}} %
\newcommand{\supp}{\mathsf{supp}} %
\newcommand{\eps}{\epsilon} %

\newtheorem{thm}{Theorem} %
\crefname{thm}{Theorem}{Theorems} 
\Crefname{thm}{Theorem}{Theorems} 

\newtheorem{prop}[thm]{Proposition} %
\crefname{prop}{Proposition}{Propositions} 
\Crefname{prop}{Proposition}{Propositions} 

\newtheorem{claim}[thm]{Claim} %
\crefname{claim}{Claim}{Claims} 
\Crefname{claim}{Claim}{Claims} 

\newtheorem{lem}[thm]{Lemma} %
\crefname{lem}{Lemma}{Lemmas} 
\Crefname{lem}{Lemma}{Lemmas} 

\crefname{conj}{Conjecture}{Conjectures} 
\Crefname{conj}{Conjecture}{Conjectures} 

\newtheorem{coro}[thm]{Corollary} %
\crefname{coro}{Corollary}{Corollaries} 
\Crefname{coro}{Corollary}{Corollaries} 

\theoremstyle{definition} %
\crefname{examp}{Example}{Examples} 
\Crefname{examp}{Example}{Examples} 

\newtheorem{defn}[thm]{Definition} %
\crefname{defn}{Definition}{Definitions} 
\Crefname{defn}{Definition}{Definitions} 

\newtheorem{remark}[thm]{Remark} %
\crefname{remark}{Remark}{Remarks} 
\Crefname{remark}{Remark}{Remarks} 

\crefname{constr}{Construction}{Constructions} 
\Crefname{constr}{Construction}{Constructions} 

\crefformat{equation}{(#2#1#3)} 
\Crefformat{equation}{(#2#1#3)} 

\newtheorem*{prop*}{Proposition} %

\newcommand{\zo}{\{0,1\}} %
\newcommand{\poly}{\mathsf{poly}} %
\newcommand{\Ext}{{\mathsf{Ext}}} %
\newcommand{\Cond}{{\mathsf{Cond}}} %

\newcommand{\Co}{\cC_{\mathsf{out}}} %
\newcommand{\Ro}{R_{\mathsf{out}}} %
\newcommand{\Ri}{R_{\mathsf{in}}} %
\newcommand{\Ci}{\cC_{\mathsf{in}}} %
\newcommand{\ens}{\mathfrak{C}}
\newcommand{\ensIn}{\mathfrak{C}_{\mathsf{in}}}

\newcommand{\vecv}{\overline{v}}

\newcommand{\Crow}{\cC_{\mathsf{row}}}
\newcommand{\Rrow}{R_{\mathsf{row}}}
\newcommand{\drow}{\delta_{\mathsf{row}}}

\newcommand{\Rcol}{R_{\mathsf{col}}}
\newcommand{\dcol}{\delta_{\mathsf{col}}}

\iftoggle{full}{}{ %
  \titleformat{\part} {\normalfont\fontsize{14}{15}\bfseries}{Part
    \thepart}{1em}{} %
  \titleformat{\section}
  {\normalfont\fontsize{12}{15}\bfseries}{\thesection}{1em}{} %
  \titleformat{\subsection}
  {\normalfont\fontsize{11}{15}\bfseries}{\thesubsection}{1em}{} %
  \titleformat{\subsubsection}
  {\normalfont\fontsize{11}{15}\bfseries}{\thesubsubsection}{1em}{} %
} %
  
\begin{document}

\pagenumbering{roman} %
\maketitle %


\begin{abstract}
  We construct constant-sized ensembles of linear error-correcting
  codes over any fixed alphabet that can correct a given fraction of
  adversarial erasures at rates approaching the Singleton bound
  arbitrarily closely.  We provide several applications of our
  results:
  \begin{enumerate}
  \item Explicit constructions of strong linear seeded symbol-fixing
    extractors and lossless condensers, over any fixed alphabet, with
    only a constant seed length and optimal output lengths;

  \item A strongly explicit construction of erasure codes on bipartite
    graphs (more generally, linear codes on matrices of arbitrary
    dimensions) with optimal rate and erasure-correction trade-offs;

  \item A strongly explicit construction of erasure codes on
    non-bipartite graphs (more generally, linear codes on symmetric
    square matrices) achieving improved rates;

  \item A strongly explicit construction of linear nearly-MDS codes
    over constant-sized alphabets that can be encoded and decoded in
    quasi-linear time.
  \end{enumerate}
\end{abstract}

\providecommand{\keywords}[1] { {\small \textbf{\textit{Keywords---}}
    #1} } 

\pagenumbering{arabic} %

{\small \tableofcontents}

\section{Introduction} \label{sec:intro}

\subsection{Background}

Erasure-correcting codes are among the most fundamental subjects of
study in classical coding theory. The erasure correction of a code is
exactly determined by its minimum distance. Over sufficiently large
alphabets, Reed-Solomon codes (more generally, MDS codes) are able to
recover from any $\delta \in [0,1)$ fraction of erasures at a rate
above $1-\delta$. This trade-off, known as the Singleton bound, is the
best to hope for \cite{MS77}. Over fixed alphabets, however, the bound
is unattainable by numerous known rate upper bounds such as the
Plotkin bound and linear programming bounds (of particular interest is
the binary alphabet).

To address the shortcoming and go around the rate vs.\ distance
trade-off barrier, relaxations of erasure correction can be considered
while still allowing a capacity bound comparable to the Singleton
bound. One is to consider random erasures, and this is Shannon's model
for the erasure channel \cite{CT06}, for which the capacity is
$1-\delta$, where, this time, $\delta$ is the expected fraction of
erasures.  Alternatively, to maintain the adversarial nature of
erasures, a usual approach is to resort to list decoding and allowing
recovery up to a small list of potential codewords (cf.\
\cite{Gur04}*{Chapter~10}).

On the other hand, it is possible to introduce a slack parameter and
require approaching the Singleton bound arbitrarily closely. In this
case, the celebrated result of Alon, Edmonds, and Luby \cite{AEL95},
known as the AEL construction, achieves explicit codes that only
require an alphabet size depending on the gap to capacity.  Namely, to
achieve a rate of at least $1-\delta-\eta$, they achieve an alphabet
size of $\exp(\tilde{O}(1/\eta^4))$ and linear time encoding and
erasure decoding.  Algebraic geometry codes, on the other hand,
achieve an exponentially better alphabet size of
$O(1/\eta^2)$. However, they are technically much more involved, and
while polynomial time constructible, known constructions, encoding,
and erasure decoding algorithms for these codes are far from achieving
nearly-linear time.  The probabilistic method achieves the trade-off
given by the Gilbert-Varshamov bound, which translates to an alphabet
size of $\exp(O(1/\eta))$ for a gap to capacity of $\eta$.

When the alphabet size $q$ is a fixed parameter, such as $q=2$, an
alternative model to allow approaching the Singleton bound for worst
case erasures is to introduce randomness in the code itself;
equivalently, to consider a family of erasure codes such that any
adversarially picked erasure pattern consisting of up to $\delta$
fraction of the positions can be corrected by almost all codes in the
family\footnote{%
  We contrast this with the so-called Monte Carlo constructions of
  codes, which corresponds to a single explicit code with a randomized
  encoder (e.g., \cite{GS16} for the error model with bit
  flips). These constructions are generally allowed to use an ample
  (e.g., $\Omega(n)$) amount of randomness in the encoder, but, on
  other hand, in some cases allow the adversary to have a controlled
  form of dependence on this randomness. As such, this is an
  incomparable model with code ensembles.  }.  This model does not
appear to have been studied as extensively in the literature. The
probabilistic method for rate vs.\ distance tradeoffs of codes of
length $n$ can be recast in this model when the size of the family is
$\exp(O(\exp(n)))$ (for nonlinear codes) or $\exp(O(n^2))$ (for random
linear codes).  Classical code ensembles such as the Wozencraft
ensemble reduce the size of the family to $\exp(O(n))$
\cites{Mas63,Jus72}.  Cheraghchi \cite{Che09} (see also
\cite{Che10}*{Chapter~5}) establishes a connection between erasure
code families and randomness extractors and condensers to construct
families of polynomial size in $n$ of comparable quality.

Motivated by various applications (e.g., distributed storage systems,
fault-tolerant hardware, among others), different erasure models that
restrict the structure of erasures have been studied in the
literature. A recently studied model of codes on graphs has been
introduced by Yohananov and Yaakobi \cite{yy19} and Yohananov, Efron,
and Yaakobi \cite{yey20}. In this model, the codewords are undirected
unweighted graphs over $N$ vertices, and the adversary erases all
edges adjacent to any $\delta$ fraction of the vertices.  In terms of
the adjacency matrix, the goal is to design a code over symmetric and
zero-diagonal $N \times N$ matrices such that for any set
$S \subseteq [N]$ of size at most $\delta N$, all codewords can be
uniquely recovered even if an adversary erases all rows and columns of
the corresponding matrix that are picked by $S$. This turns out to be
a special case of a more generalized framework defined in
\cite{alo23}. The notion can naturally be extended to non-binary
alphabets as well, and of special interest are linear codes over
$\F_q^{N \times N}$ with such properties.

Interestingly, if the code is linear over $\F_2$ and contains the
all-ones codeword, this means that all other non-zero codewords define
Ramsey graphs (i.e., contain no cliques or independent sets of size at
least $\delta N$). Therefore, the problem can also be regarded as the
packing of an exponentially large collection of pairwise-distant
Ramsey graphs.

Over alphabets of size at least $N$, the tensor product of an MDS code
with itself can lead to such codes at a rate larger than
$(1-\delta)^2$, which is the best to hope for.  Similarly, for random
row and column erasures over any alphabet, including binary (e.g.,
Shannon's model), two copies of an off-the-shelf linear
capacity-achieving code can be tensored together to achieve the
desired erasure correction.  Unlike classical erasure codes, however,
this bound can also be attained arbitrarily closely for adversarial
erasure patterns over any fixed alphabet, including binary.  This can
be confirmed using the probabilistic method by analyzing random linear
graph codes \cite{KPS24}*{Proposition~3.1}.  Explicit constructions of
this quality, however, are much more challenging to achieve over small
alphabets. For the binary alphabet and constant erasure fraction
parameter $\delta \in [0,1)$, the state of the art trade-offs achieved
by explicit constructions are $R=1-2\delta$ (when $\delta < 1/2$)
\cite{yy19} and, incomparably, $R=(1-\delta^{1/3})^6 - o(1)$
\cite{KPS24} for all $\delta$. If we further require strong
explicitness\footnote{See \cref{sec:prelim} for the definition.}, the
best known construction only achieves $R=(1-\delta^{1/4})^8 - o(1)$
\cite{KPS24}. These constructions are significantly far from the
optimal trade-off of $R=(1-\delta)^2-o(1)$.

A closely related notion is when the code is over $M\times N$ matrices
over a $q$-ary alphabet and the requirement is recovery against any
$\drow$ fraction of row erasures and $\dcol$ fraction of column
erasures, for parameters $(\drow,\dcol)\in [0,1)^2$. In this case, $M$
and $N$ need not be equal and the set of row and column erasures can
be independent.  When $q=2$, this can be thought of as a bipartite
variation of graph codes with $M$ and $N$ vertices on either side. The
notion can also be captured by that of directed graph codes defined in
\cite{KPS24}*{Definition~4.2} and is related to the crisscross error
model on matrix codes, as studied in \cite{criss97}.  In this case,
the capacity of the model (i.e., the best rate to hope for) becomes
$(1-\drow)(1-\dcol)$, which can be attained arbitrarily closely over
any alphabet by random linear bipartite graph codes. Although this
model is slightly different from the non-bipartite graph codes
mentioned before, to the best of our knowledge, the best explicit
constructions of bipartite graph codes only attain the same trade-offs
as those for non-bipartite graph codes (in which case, $M=N$).

\subsection{Overview of the Results and
  Techniques}\label{sec:overview}

\paragraph{Linear Erasure Code Family.}
One of the main technical tools that we develop in this work is an
explicit construction of linear erasure code families of
\emph{constant size} over any fixed alphabet.  More precisely, for any
erasure parameter $\delta \in [0,1)$, any finite field $\F_q$, and
arbitrarily small slack parameter $\eta$ and error parameter $\eps>0$,
we construct a family of $\F_q$-linear codes of a desired length $N$
and near-optimal rate; namely, at least $1-\delta-\eta$, such that any
erasure pattern of up to $\delta N$ positions can be corrected by at
least $1-\eps$ fraction of the codes in the family.  Moreover, the
size of the code ensemble is independent of $N$ and only
(polynomially) depends on the parameters $1/\eta$ and $1/\eps$.  In
fact, our construction is \emph{strongly explicit} in that each entry
of a generator matrix for each code in the family can be computed in
polynomial time in $\log N$. Encoding and erasure decoding can both be
performed in quasi-linear (i.e., $\tilde{O}(N)$) time.  This is
achieved by a code concatenation technique and a randomness-efficient
permutation of the coordinates using a randomness extractor.  In
contrast, the result of \cite{Che09} constructs an erasure code family
of polynomial size in $N/\eps$ and exponential in $1/\eta$.  A
simplified statement of our result is recorded below.
\begin{thm}[\Cref{coro:explicit},
  Simplified] \label{thm:codes:simplified} For any $\delta \in [0,1)$,
  $\eta>0$, prime power $q$, and large enough $N$, there is a strongly
  explicit construction of an ensemble of linear codes of length $N$
  over $\F_q$ of rate at least $1-\delta-\eta$ such that any pattern
  of up to $\delta N$ erasures can be corrected by all but up to an
  $\eta$ fraction of the codes in the ensemble. The code ensemble is
  of size $\poly(1/\eta)$.
\end{thm}

\paragraph{Extractors for Symbol-Fixing Sources.}
The correspondence between erasure code families and randomness
extractors in \cite{Che09} can turn the above construction into an
explicit, seeded, linear, and strong randomness extractor (or lossless
condenser) for (oblivious) symbol-fixing sources over $\F_q$ with only
a constant seed length (see \cref{sec:ext:facts} for
background). Namely, we prove the following.

\begin{coro}[\Cref{coro:explicit:bitfixing},
  Simplified] \label{coro:explicit:bitfixing:simplified} For any
  $\delta \in [0,1)$, $\eta>0$, prime power $q$, and large enough $N$,
  there are explicit constructions of functions
  $\Ext\colon \F_q^N \times \zo^d \to \F_q^{(\delta-\eta) N}$ and a
  $\Cond\colon \F_q^N \times \zo^d \to \F_q^{(\delta+\eta) N}$ where
  $d=O(\log(1/(\eps \eta))$.  The functions $\Ext$ and $\Cond$ are a
  strong linear $(\delta N, \eps)$-extractor and linear
  $(\leq \delta N, \eps)$-lossless condenser, respectively, for
  symbol-fixing sources.  \qedhere \qed
\end{coro}

This is rather surprising, especially considering that for the
slightly more general class of affine sources, the probabilistic
method can only show the existence of strong seeded linear extractors
with seed length not much better than what extractors for general
sources can achieve.  Moreover, the probabilistic method shows that
there are seedless extractors for symbol-fixing and affine sources
that can extract almost all entropy. By now, several explicit
constructions approaching this goal are also known (e.g.,
\cites{CGL22,DF25,Li16,Rao09}, among many others).  However, it is
important to note that such functions are fundamentally nonlinear. In
particular, no seedless linear extractors can exist for affine
sources. Furthermore, seedless linear extractors for bit-fixing
sources are restricted by the rate vs.\ distance trade-offs of codes
and cannot extract all entropy. In fact, the Plotkin bound implies
that no fixed linear function can extract more than a constant
\emph{number of bits} from $q$-ary symbol-fixing symbols with an
entropy rate of at most $1/q$. Moreover, there is a positive entropy
rate (only depending on $q$) below which the only possible linear
symbol-fixing extractor simply adds up the input symbols over $\F_q$
(and thus can only extract one $\F_q$ symbol).

Symbol-fixing (with the important special case of bit-fixing)
extractors have been extensively studied in the literature.  Along
with the more general notion of affine extractors, they have versatile
applications in pseudorandomness (e.g., \cites{Gab10,KJS01,CZ19}),
information-theoretic cryptography (cf.\ \cite{Dod00}), complexity
theory (e.g., \cite{HIV22}), and algorithms (e.g., \cite{CI17}).
Several notions in various applications turn out to be either
equivalent or closely related to symbol-fixing extractors, such as
(exposure-) resilient functions (cf.\
\cites{Dod00,CGHFRS85,Fri92,Sti93,FT00}), all-or-nothing transforms
(AONT) \cites{Riv97,CDHKS00}, threshold secret sharing schemes (see
\cite{LCGSW19}), and wiretap codes \cite{CDS11}.  Several such
applications (e.g., \cites{CDS11,CI17,LCGSW19}) crucially require
affine and symbol-fixing extractors that are linear (possibly allowing
a seeded).

\paragraph{Optimal Bipartite Graph Codes.} We use our explicit
construction of erasure code families of optimal rate to provide a
strongly explicit construction of capacity-achieving erasure codes
over bipartite graphs (whether the number of left and right vertices
are equal or not). In other words, for any fixed prime power $q$, we
provide a strongly explicit construction of linear codes over
$\F_q^{M \times N}$ that can recover from any $\drow$ fraction of row
erasures and any $\dcol$ fraction of column erasures, achieving
optimal rate $(1-\drow)(1-\dcol)-o(1)$.  Techniques in prior work
\cites{yy19,KPS24} can be adapted to this case and provide rates
$1-\drow-\dcol$ \cite{yy19} and
$(1-\drow^{1/3})^3(1-\dcol^{1/3})^3-o(1)$ \cite{KPS24} for explicit
constructions, and $(1-\drow^{1/4})^4(1-\dcol^{1/4})^4-o(1)$ if strong
explicitness is desired \cite{KPS24}.  The following is a simplified
statement of our result.
\begin{coro}[\Cref{coro:matrix:explicit}, Simplified]
  \label{coro:matrix:explicit:simplified}
  For any $(\drow, \dcol) \in [0, 1)^2$, and large enough $N$ and $M$
  ($M$ being a power of two), there is a strongly explicit
  construction of a linear code over $\F_q^{M \times N}$ achieving
  rate at least $(1-\drow)(1-\dcol)-o(1)$ that can recover from any
  $\drow$ fraction of row erasures and any $\dcol$ fraction of column
  erasures.  The code can be encoded and erasure-decoded in
  quasi-linear time. \qedhere \qed
\end{coro}

\paragraph{Nearly-MDS Codes over Constant Alphabet.}
As a consequence of our construction for the special case when no row
erasures can occur (and $M \ll N$), we obtain an AEL-type linear code
construction. Namely, for any fixed $\F_q$ and $\delta \in [0,1)$, and
for a gap parameter $\eta > 0$, we provide a strongly explicit
construction of $\F_q$-linear codes that can correct any $\delta$
fraction of erasures at a rate of at least $1-\delta-\eta$. As is the
case for the AEL construction \cite{AEL95}, our codes achieve a
constant alphabet size that only depends on the gap to capacity
$\eta$.  This is recorded below.

\begin{thm}[\Cref{thm:almost:MDS:improve},
  Simplified] \label{thm:almost:MDS:improve:simplified} For any
  $\delta \in [0,1)$, prime power $q$, parameter $\eta>0$, and large
  enough $N$, there is a strongly explicit construction of an
  $\F_q$-linear code achieving relative distance larger than $\delta$
  and rate at least $1-\delta-\eta$ over an alphabet of size
  $\exp(\poly(1/\eta))$.  Furthermore, the code can be encoded and
  erasure decoded in quasi-linear time.  \qedhere \qed
\end{thm}

Using existing constructions of randomness extractors, the resulting
alphabet size is bounded by $\exp(\tilde{O}(1/\eta^8))$.  
Assuming nearly optimal extractors, this can be improved to 
$\exp(\tilde{O}(1/\eta^4))$, matching the alphabet size achieved 
by \cite{AEL95}.

Since \cite{AEL95} is based on expander codes \cites{SS96,Spi95} that
are defined either via the parity check matrix or a layered
construction in systematic form, we are unable to verify the strong
explicitness of this construction, whereas our codes are constructed
with strong explicitness in mind. We demonstrate that our paradigm of
using erasure code families to construct codes arbitrarily approaching
the Singleton bound can potentially achieve an alphabet size of
$\exp(O(1/\eta^2))$ assuming explicit construction of erasure code
family with optimal size; thus an exponent which is quadratically
better than what \cite{AEL95} can achieve even using optimal
(Ramanujan) expanders (see also \cref{sec:ext:facts} for a discussion
on algebraic geometry codes that exponentially outperform random
codes).

\paragraph{Non-Bipartite Graph Codes.}
Finally, we construct strongly explicit erasure codes over symmetric
square matrices over $\F_q$ with zero diagonals.  This corresponds to
erasure codes over non-bipartite graphs. More precisely, we obtain the
following.

\begin{thm}[\Cref{thm:graph:code:explicit}, Simplified]
  \label{thm:graph:code:explicit:simplified}
  For any $\delta \in [0, 1)$, prime power $q$, and large enough $N$,
  there is a strongly explicit construction of a linear code over
  symmetric matrices with zero diagonals in $\F_q^{N \times N}$
  achieving rate at least $(1-\sqrt{\delta})^4-o(1)$.  Furthermore,
  the code can be encoded and erasure-decoded against any $\delta$
  fraction of row and column erasures in quasi-linear time.  \qedhere
  \qed
\end{thm}

This improves the previously known strongly explicit constructions
achieving rates $(1-\delta^{1/4})^8-o(1)$ and explicit constructions
with rate $(1-\delta^{1/3})^6-o(1)$, both from \cite{KPS24}, for the
whole range of $\delta$.  The construction follows a framework similar
to \cite{KPS24} that concatenates a tensor code and a nearly optimal
bipartite graph code. Our improvement is mainly a consequence of
better choices of the underlying codes, which we also construct in
this paper. Concretely, we use the tensor product of two copies of a
strongly explicit nearly-MDS code, over a constant-sized alphabet, as
the outer code. For the inner code, we use our strongly explicit
construction of optimal bipartite graph codes. In contrast,
\cite{KPS24} resorts to an exhaustive search for a suitable inner
code. This allows us to apply a single concatenation, unlike the
iterative concatenation that is needed in \cite{KPS24}, leading to our
improvement.

\subsection{Organization} \label{sec:org} The rest of the article is
organized as follows. In \cref{sec:prelim}, we recall basic
definitions and notation that are used
throughout. \Cref{sec:ext:facts} recalls the notion of symbol-fixing
extractors and establishes their connections to erasure code families
that are also formalized in this section.  \Cref{sec:codes} provides
the main technical tool used throughout the work; namely, a strongly
explicit construction of erasure code families of constant size over
any fixed alphabet.  \Cref{sec:graphs} uses this to provide a strongly
explicit construction of optimal linear erasure codes over bipartite
graphs (equivalently, over matrices of desired dimensions over any
$\F_q$).  As an immediate consequence of this construction, strongly
explicit codes arbitrarily achieving the Singleton bound over
constant-sized alphabets (depending on the gap to capacity) are
constructed in \cref{sec:AEL}. Finally, \cref{sec:symmetric} provides
a strongly explicit construction of erasure codes over non-bipartite
graphs (equivalently, over symmetric matrices over $\F_q$) that
improve the state of the art on the rate versus erasure correction
trade-off for explicit codes.

\subsection{Preliminaries and Notation}\label{sec:prelim}

Throughout the paper, $q$ is a fixed prime power. Of particular
interest is the binary case where $q=2$.  For any matrix $M$ whose
rows and columns are indexed by $A$ and $B$, respectively, and any
subset $S\subseteq A$ of rows and any subset $T\subseteq B$ of
columns, we use $M|_{S,T}$ to denote the submatrix of $M$ consisting
of the rows and the columns indexed $S$ and $T$, respectively. For any
specific row index $a\in A$ and column index $b\in B$, we use $M[a,b]$
to denote the entry of $M$ at row indexed by $a$ and column indexed by
$b$.  An explicit construction of a linear code is one that is
equipped with an algorithm that outputs a generator matrix $G$ for the
code over the underlying field in polynomial time in the size of the
matrix. We say that the construction is \emph{strongly explicit} if
there is an algorithm that, given a row and a column index, outputs
the corresponding entry in $G$ (as an element of the underlying field)
in polynomial time in the bit-length of the indices (and output
length).  The set $\{1, \ldots, N\}$ is denoted by $[N]$.  All
logarithms are taken to base $2$.  We occasionally use the asymptotic
notation $\tilde{O}(f(n))$ as a shorthand for
$O(f(n) (\log f(n))^{O(1)})$.  We also use the notation $\poly(f(n))$
for $f(n)^{O(1)}$.

\paragraph{Extractors and Condensers.}
Let $\Omega$ be a finite set and $X$ be a distribution defined by the
probability mass function $p_X\colon \Omega \to \R^{\geq 0}$. The
min-entropy of $X$, denoted by $H_\infty(X)$ is defined as
\[
  H_\infty(X) \coloneq \min_{x \in \supp(X)} -\log p_X(x),
\]
where $\supp(X)$ denotes the support of $X$; i.e., the set of outcomes
with non-zero probability mass.  We use $U_\Omega$ to denote the
uniform distribution on $\Omega$.  The $\ell_1$ distance between two
probability measures $p_X$ and $p_Y$ over $\Omega$ is defined as the
usual geometric $\ell_1$ distance when the distributions are regarded
as vectors of probabilities; namely,
$\|p_X-p_Y\|_1 \coloneq \sum_{x\in \Omega} |p_X(x) - p_Y(x)|$. This is
twice the statistical (or total variation) distance between the two
distributions.  When there is no risk of confusion, we may refer to a
random variable to imply its underlying probability distribution. Two
distributions $X$ and $Y$ are $\eps$-close if their statistical
distance is at most $\eps$.  This is denoted as $X \sim_\eps Y$. We
use $X \sim Y$ to denote that the random variable $X$ is drawn from
the distribution $Y$.  Overloading the notation, for a set $\Omega$,
we use the shorthand $X \sim \Omega$ for $X \sim U_\Omega$; i.e., $X$
is uniformly sampled from $\Omega$.

For finite sets $Z$ and $\Omega'$, a function
$\Ext\colon \Omega \times [D] \to \Omega'$ is a (strong, seeded)
$(k, \eps)$-extractor if, for any random variable $X$ on $\Omega$ with
$H_\infty(X) \geq k$, and an independent $Z \sim [D]$, the
distribution of $(Z, \Ext(X,Z))$ is $\eps$-close to the uniform
distribution over $[D] \times \Omega'$. By an averaging argument, this
implies that for any $\eps_1 \eps_2 = \eps$, for all but at most an
$\eps_1$ fraction of seeds $z \in [D]$, the function $\Ext(\cdot, z)$
extracts the source $X$ within error $\eps_2$ (i.e.,
$\Ext(X, z) \sim_{\eps_2} \U_{\Omega}$).

In this work, we shall use explicit constructions of strong explicit
extractors for the high min-entropy regime. In order to capture any
future progress on the state of the art for extractor constructions,
we provide an abstract formulation of the guarantees that we need
below.

\begin{defn} \label{defn:ext:constr} For absolute constants
  $\gamma_1,\gamma_2 \geq 2$, we say that extractors are
  $(\gamma_1,\gamma_2)$-attainable (resp., strongly
  $(\gamma_1,\gamma_2)$-attainable) if the following holds for some
  function $f(\Delta)$.  For any fixed $\Delta > 0$, large enough $n$,
  and error parameter $\eps > 0$, there is a strong
  $(n-\Delta, \eps)$-extractor
  $\Ext\colon \zo^n \times \zo^d \to \zo^m$ where
  $d \leq f(\Delta) + \gamma_1 \log(1/\eps) + O(1)$ and
  $m \geq n-\Delta - \gamma_2 \log(1/\eps) - O(1)$.  Moreover, $\Ext$
  runs in polynomial time in $2^n$ (resp., polynomial time in
  $n/\eps$).
\end{defn}

The probabilistic method shows that (without considering the runtime),
the above definition can be satisfied for $\gamma_1=\gamma_2=2$ for
$f(\Delta) = \log(\Delta)$ \cite{AB09}*{Section~21.5.4} and that this
is the best to hope for \cites{NZ96,RTS00}.

As for explicit constructions, below we quote an explicit construction
of extractors based on the zig-zag product of graphs.

\begin{thm}\cite{RVW00}*{Rephrased} \label{thm:RVW}
  For any $\Delta > 0$, there is a strong\footnote{The proof details
    of the claim that the extractor is strong appears in the full
    version of this work \cite{RVW01}*{Remark~6.8}.}
  seeded $(n-\Delta, \eps)$-extractor
  $\Ext\colon \zo^n \times \zo^d \to \zo^m$ where
  $d = 2 \log \Delta + 4 \log(1/\eps) + O(1)$ and
  $m = n - \Delta - 2 \log(1/\eps) - O(1)$.  Moreover, the function
  can be computed in time $2^{2^{O(\Delta)}} \cdot \poly(n)$. \qedhere
  \qed
\end{thm}

In the language of \cref{defn:ext:constr}, this immediately translates
into the following.

\begin{prop} \label{prop:ext:constr} Extractors are strongly
  $(4,2)$-attainable.
\end{prop}

Among other extractors that fit our parameter regimes are those
constructed in \cites{GW97,CRVW02} which achieve larger $\gamma_1$
(and, moreover, \cite{CRVW02} is only weakly explicit in the sense of
running in polynomial time in $2^n$).

A dual notion to strong extractors is that of lossless condensers.  A
(strong, seeded) $(\leq k,\eps)$-lossless condenser is a seeded
function $\Cond\colon \Omega \times D \to \Omega'$ such that for any
random variable $X$ on $\Omega$ with $H_\infty(X) \leq k$, and an
independent $Z \sim [D]$, the distribution of $(Z,\Cond(X,Z))$ is
$\eps$-close to a distribution of min-entropy $H_\infty(X)$.

Linear extractors (resp., lossless condensers) are seeded extractors
(resp., lossless condensers) that are linear functions of their inputs
for each fixed seed. Other than general sources, of particular
interest to us is the class of (oblivious) bit-fixing and
symbol-fixing sources.  A symbol-fixing source (also called bit-fixing
for the special case of $q=2$) of ($q$-ary) entropy $k$ is a
distribution over $\F_q^N$ where $k$ of the $N$ coordinates are
uniform and independent over $\F_q$ and the rest are fixed to
arbitrary values. They are a special case of the more general affine
sources that are defined by the uniform distribution over a subspace
of $\F_q^N$ of dimension $k$. A linear algebraic argument implies that
if a linear function extracts an affine source within any error less
than $1/2$, the error must actually be exactly zero. Consequently, a
linear seeded affine (and, in particular, symbol-fixing) extractor
with error at most $\eps$ must extract any affine source (of
sufficient entropy) perfectly (i.e., with zero error) for all but at
most an $\eps$ fraction of seeds.

\paragraph{Graph Codes.}
In \cref{sec:graphs,sec:symmetric}, we provide explicit constructions
of linear bipartite and non-bipartite graph codes that are defined
here.

\begin{defn} \label{def:bigraph:code} An $[M,N,\drow,\dcol]_q$-graph
  code is a code over $M \times N$ matrices with entries from a
  $q$-ary alphabet such that no two codewords coincide on any
  submatrix containing at least $(1-\drow)M$ of the rows and at least
  $(1-\dcol)N$ of the columns.  When $q$ is a prime power, the code is
  linear when it is a linear subspace of $\F_q^{M \times N}$.
\end{defn}

Note that when $q=2$, each codeword can be thought of as the adjacency
matrix of a bipartite graph with $M$ left vertices and $N$ right
vertices. In that case, any graph from the code can be uniquely
identified if an adversary erases all edges adjacent to at most a
$\drow$ fraction of the left vertices and at most a $\dcol$ fraction
of the right vertices. Alternatively, one can interpret the codewords
as adjacency matrices of directed graphs \cite{KPS24}. A related term
in the literature is the notion of matrix codes against crisscross
erasures \cite{criss97}. The notion of rate for codes described by
\cref{def:bigraph:code} is defined in the standard sense.

For non-bipartite graphs, the analogous definition consists of square
symmetric matrices, recorded below.

\begin{defn} \label{def:graph:code} An $[N,\delta]_q$-graph code $\cC$
  is a code over $N \times N$ symmetric matrices with entries from a
  $q$-ary alphabet (and rows and columns indexed by $[N]$) and
  all-zeros diagonals (``zero'' being any fixed element of the
  alphabet) such that, for any set $S \subseteq [N]$ of size at most
  $\delta N$, no two codewords coincide on the submatrix picked by the
  set of rows and columns that lie in $[N]\setminus S$.  When $q$ is a
  prime power, the code is linear when it is a linear subspace of
  $\F_q^{N \times N}$.  The rate of the code is defined to
  be\footnote{%
    This is defined so that the rate of the identity code becomes
    $1$. In terms of the analogy with undirected graphs, there are
    $\binom{n}{2}$ possible undirected graphs, and the rate measures
    the density of a packing of graphs.}  $\log_q|\cC|/\binom{n}{2}$.
\end{defn}

When $q=2$, each codeword can be thought of as the adjacency matrix of
an undirected non-bipartite graph with $N$ vertices.  Any graph from
the code can then be uniquely identified if an adversary erases all
edges adjacent to at most a $\delta$ fraction of the vertices. We
recall the achievability result in \cite{KPS24} on random linear graph
codes below.

\begin{prop}\cite{KPS24}*{Proposition 3.1}
  \label{prop:nonbipartite:existence}
  For any prime power $q$, fixed $\delta \in [0,1)$ and parameter
  $\eta>0$, there is an $N_0 = O(1/\eta)$ such that the following
  holds. For all $N\ge N_0$, there exist linear $[N,\delta]_q$-graph
  codes having rate at least $(1-\delta)^2-\eta$. \qedhere \qed
\end{prop}

It is straightforward to observe from
\cref{def:bigraph:code,def:graph:code} that any
$[M,N,\drow,\dcol]_q$-graph code (resp., $[N,\delta]_q$-graph code)
must have rate at most $(1-\drow)(1-\dcol)$ (resp., $(1-\delta)^2$) by
simply considering a single erasure pattern.  Therefore, the result of
\cref{prop:nonbipartite:existence} can be seen as a characterization
of the ``capacity'' of this erasure model.  This can be adapted to the
case of bipartite graph codes as we demonstrate below. In fact, in
this work, we show the achievability of the capacity with a strongly
explicit construction.

\begin{prop} \label{prop:bipartite:existence} For any prime power $q$,
  fixed $(\drow, \dcol) \in [0,1)^2$ and parameter $\eta>0$, there is
  an $M_0 = O(1/\eta)$ such that the following holds. For all integers
  $M,N$ satisfying $\min\{M,N\}\geq M_0$, there exist linear
  $[M,N,\drow,\dcol]_q$-graph codes having rate at least
  $(1-\drow)(1-\dcol)-\eta$.
\end{prop}
\begin{proof}
  Our proof closely follows the proof of
  \cite{KPS24}*{Proposition~3.1}.  Let
  $k \coloneq \lfloor(1-\drow)(1-\dcol)-\eta\rfloor MN$. Sample $k$
  uniformly random and independent matrices
  $G_1,\ldots, G_k \in \F_q^{M \times N}$.  For a vector
  $\vecv \in \F_q^k$, denote
  \[
    G_{\vecv} \coloneqq \sum_{i \in k} \vecv(i)G_k.
  \]
  We take the set $\{G_{\vecv} \mid \vecv \in \F_q^k \}$ to be our
  random graph code. Note that this graph code is $\F_q$-linear, and
  therefore in order to show that it can, with high probability,
  recover from any $\drow$ fraction of row erasures and $\dcol$
  fraction of column erasures, it suffices to show that for every
  erasure pattern and every non-zero codeword, at least one non-zero
  entry survives after applying the erasure pattern to the codeword.

  Fix some non-zero $\vecv$. Observe that $G_{\vecv}$ is a uniformly
  random matrix.  For some $S \subseteq [M]$ and $T \subseteq [N]$,
  denote by $E(G_{\vecv}, S, T)$ the undesirable event where, upon
  erasing the rows and columns of $G_{\vecv}$ indicated by $S$ and $T$
  respectively, every unerased entry is equal to zero. Denote by
  $E(G_{\vecv})$ the event where $E(G_{\vecv}, S, T)$ holds for at
  least one pair $(S,T)$ satisfying $|S| = \drow M$ and
  $|T| = \dcol N$. Then,
  \begin{align*}
    \Pr[E(G_{\vecv})] &\le \sum_{{\substack{S \subseteq [M], T
                        \subseteq [N],\\
    |S| = \drow M, |T| = \dcol N}}} \Pr[E(G_{\vecv}, S, T)] \\
                      &\le {\binom{M}{\drow M}}{\binom{N}{\dcol N}}
                        \cdot q^{-MN(1-\drow)(1-\dcol)} \\
                      &\le 2^{h_2(\drow)M+h_2(\dcol)N}
                        \cdot q^{-MN(1-\drow)(1-\dcol)} \\
                      &\le q^{h_2(\drow)M+h_2(\dcol)N-MN(1-\drow)(1-\dcol)} \\ 
                      &\le q^{M+N-MN(1-\drow)(1-\dcol)}.
  \end{align*}
  Here, $h_2(x) \coloneq -x \log_2 x -(1-x)\log_2(1-x)$ is the binary
  entropy function, and we have used the well-known inequality
  $\binom{a}{b} \leq 2^{a h_2(b/a)}$.  Upon applying the union bound
  over all non-zero $\vecv \in \F_q^k$:
  \begin{align*}
    \Pr\left[ \bigcup_{\vecv \in \F_q^k \setminus \{\mathbf{0}^k\}}
    E(G_{\vecv})\right] &\le q^k \cdot q^{M+N-MN(1-\drow)(1-\dcol)}
                          \le q^{-\frac{\eta MN}{2}}.
  \end{align*}
  The last inequality is true as long as
  $k \le ((1-\drow)(1-\dcol)-\eta)MN$ and $M+N < \eta MN/2$. The
  latter inequality can be ensured by taking $M$ and $N$ large enough;
  e.g., $\min\{M,N\} \geq M_0$ for some $M_0 = O(1/\eta)$.

  The union bound implies that, with high probability, the random
  construction can withstand all erasure patterns of concern (and,
  therefore, any smaller erasure patterns as well). In particular, by
  considering empty erasure patterns, we also have proven that (with
  high probability) all codewords are non-zero.  Therefore, the
  quantity $k/MN$ is indeed the rate of the code.
\end{proof}

\section{Erasure Code Families and Symbol-Fixing
  Extractors} \label{sec:ext:facts}

\subsection{Erasure Code Families}

Our main object of study is the notion of erasure code families,
formally defined below.

\begin{defn}[Erasure Code Family] \label{def:erasure:family} An
  $[n,\delta,\eps]_q$-erasure code family is an ensemble of linear
  codes over $\F_q^n$ such that the following holds. For any set
  $S \subseteq [n]$ of size at most $\delta n$, all but at most an
  $\eps$ fraction of the codes in the set are able to correct the
  erasure pattern incurred by $S$. We say that the code family has
  rate $R$ if the rate of all but at most an $\eps$
  fraction\footnote{Alternatively, we could have required all codes in
    the ensemble to be of rate $R$. This can be trivially ensured by
    artificially adjusting the rate of any rate-deficient codes in the
    ensemble to be exactly $R$ and doubling the parameter $\eps$.  }
  of the codes in the family is $R$. We say that an erasure code
  family construction is explicit (resp., strongly explicit) if there
  is an algorithm that, given an index $i$, can construct the $i$th
  code in the ensemble explicitly (resp., strongly explicitly).
\end{defn}

Recall that a linear code can correct an erasure pattern $S$ if and
only if the following equivalent conditions hold.

\begin{enumerate*}[label=(\arabic*)]
\item Any generator matrix for the code with the columns indexed by
  $S$ removed has full row rank; and
\item Any parity check matrix for the code with the columns outside
  $S$ removed has full column rank.
\end{enumerate*}

Obviously, a single code over $\F_q^n$ of relative distance larger
than $\delta$ is an $[n,\delta,0]_q$-erasure code family of size
one. This can be achieved at a rate matching the Singleton bound
$1-\delta$ for large alphabets; i.e., as long as $q \geq n$ (i.e., the
Reed-Solomon or any MDS code). For small (in particular, constant)
sized alphabets; however, the rate-distance trade-offs of codes
prevents any code from approaching the Singleton bound.  However, our
goal is to show that a constant-sized family of codes can do so
instead. Of the known rate upper bounds on codes on a given alphabet
and relative distance, we recall the Plotkin bound\footnote{There are
  a variety of asymptotic bounds, such as the well-known MRRW bounds
  based on linear programming, that are tighter than the Plotkin bound
  for small (such as binary) alphabets, or all alphabets for the
  extremal distance regime. Plotkin bound, however, performs best for
  larger (including large constant) alphabets while still allowing an
  explicit asymptotic expression for the entire range of the distance
  parameter.} as follows.

\begin{thm}[Plotkin Bound] \label{thm:Plotkin} \cite{MS77}*{Chapter~2}
  Any $q$-ary code with relative distance $\delta \in [0,1-1/q)$
  achieves a rate upper bounded by $1-\delta q/(q-1)+o(1)$ where the
  $o(1)$ term vanishes as the block length grows.  For larger
  $\delta$, the upper bound on the rate is $o(1)$.
\end{thm}

Observe the immediate corollary of the Plotkin bound that any
sufficiently long code with a fixed relative distance
$\delta \in (0,1)$ and rate at least $1-\delta-\eps$ requires an
alphabet of size $\Omega(1/\eps)$.

\paragraph{Random and Algebraic Geometry Codes.}
On the existence aspect, random codes achieve the Gilbert-Varshamov
bound, which, in this setting, yields an exponential alphabet size of
$\exp(O(1/\eps))$.  Somewhat miraculously, algebraic geometry codes
are known to achieve an exponentially better alphabet of size
$O(1/\eps^2)$ compared to random codes. This is a consequence of the
result below.

\begin{thm} \label{thm:AG} \cite{TVZ82} Let $q \geq 49$ be an even
  power of a prime and $\delta \in [0,1)$.  Then, there is an explicit
  construction of $q$-ary codes of large enough length and achieving
  relative distance $\delta$ and rate at least
  $1-\delta-1/(\sqrt{q}-1)$.
\end{thm}

The above is the so-called TVZ bound, after Tsfasman-Vl{\u a}du{\c
  t}-Zink who first described such codes \cite{TVZ82}. More efficient
constructions of such codes were later obtained by Garcia and
Stichtenoth \cite{GS95} and Shum et al.~\cite{SAKS01}. The latter
construction, despite being explicit, provides only a near-cubic time
algorithm to compute a generator matrix for the code.  Therefore,
these constructions are not strongly explicit.  For any linear code,
erasure decoding can be done in nearly cubic time by Gaussian
elimination. Reed-Solomon codes can be designed to allow for
quasi-linear time erasure correction using FFT-based algorithms.
However, for the above-mentioned algebraic geometry codes, to the best
of our knowledge, no significant improvements over Gaussian
elimination erasure decoding are known.  For some classes of algebraic
geometry codes, FFT-based encoders are known \cite{LLLEWX24}.
However, for algebraic geometry codes achieving the bounds in
\cref{thm:AG}, despite the existence of sub-quadratic time encoders
\cite{NW19}, quasi-linear time encoding remains elusive due to the
difficulty of constructing an explicit basis for such codes.

\subsection{Existence of Erasure Code Families}
\label{sec:erasure:existence}
Using the probabilistic method, it is possible to verify the existence
of erasure code families as follows.

\begin{lem} \label{lem:existence:code} For any $\delta \in [0, 1)$ and
  $\eta > 0$, there is an $[n,\delta,\eps]_q$-erasure code family of
  rate $R = 1-\delta-\eta$ and size $t$, provided that
  $t \geq 2/(\eta \eps \log q)$ and $n \geq n_0$ for some
  $n_0 = O(\log(1/\eps)/(\eta \log q))$.
\end{lem}

\begin{proof}
  Our code ensemble consists of a collection of $t$ independently
  sampled random linear codes over $\F_q$.  Namely, for the given rate
  parameter $R$, each code in the ensemble is generated by a uniformly
  random $Rn \times n$ matrix over $\F_q$. We make use of the
  following well-known fact.

\begin{claim}
  Let $k \leq n$ be integers and $M \in \F_q^{k \times n}$ be drawn
  uniformly random. Then, the probability that $M$ has rank less than
  $k$ is at most $q^{k-n}$.
\end{claim}

\begin{proof}[Proof (of Claim)]
  Since each row of the matrix must avoid the span of the previous
  rows, the number of choices for $M$ of rank $k$ is as follows.
  \begin{align*}
    \prod_{i=0}^{k-1} (q^n-q^i) = q^{nk} 
    \prod_{i=0}^{k-1} (1-q^{i-n})\geq q^{nk}
    \left(1-\sum_{i=0}^{k-1} q^{i-n}\right)=
    q^{nk} \left( 1- q^{-n} \frac{q^k-1}{q-1} \right)
    \geq q^{nk}(1-q^{k-n}).
  \end{align*}
  and thus the probability of $M$ being of full row rank is at least
  \[
    \frac{q^{nk}(1-q^{k-n})}{q^{kn}} = 1 - q^{k-n}.  \qedhere
  \]
\end{proof}

Using the above claim, the chance that a random $Rn \times n$ matrix
fails to generate a code of dimension $Rn$ is at most $q^{(R-1)n}$, an
exponentially small probability.  Therefore, we can assume that all
codes in the ensemble have the same rate $R$.

Consider any erasure pattern $S \subseteq [n]$ of size at most
$\delta n$.  Recall that a linear code can correct the erasure pattern
determined by $S$ if and only if a generator matrix of the code with
the columns in $S$ removed retains a full row rank. Using the above
claim, the chance of this not being the case for a specific code in
our ensemble is at most
$\nu \coloneq q^{(R+\delta-1) n} = q^{-\eta n}$.

We are interested in the event that all but at most $\eps t$ of the
independently sampled codes in the ensemble can correct the erasure
pattern determined by $S$. It suffices to ensure that this event
occurs with a probability less than $2^{-n}$, so that a union bound on
all choices of $S$ can guarantee the existence of our desired code
ensemble.  We do so by analyzing the probability that some set $T$ of
the code ensemble of size larger than $\eps t$ cannot recover from the
erasure pattern $S$.  For a fixed $T$, this occurs with probability at
most $\nu^{|T|} \leq q^{-\eta n \eps t}$ due to the independence of
the codes in the ensemble.  We finally take a union bound on all
choices of $T$.  Altogether, it suffices to ensure that
\[
  \binom{t}{\eps t} q^{-\eta n \eps t} < 2^{-n},
\]
which holds for $t \geq 2/(\eta \eps \log q)$ as long as
$n \geq 2 \log(e/\eps)/(\eta \log q)$ (using the estimate
$\binom{a}{b} \leq (a e/b)^b$, where $e$ is the base of natural
logarithm).
\end{proof}

\subsection{Connection with Symbol-Fixing Extractors}

We recall the connection observed between erasure code families and
symbol-fixing extractors in \cite{Che09} (see also
\cite{Che10}*{Chapter~5}).

\begin{lem}\cite{Che09} \label{lem:ExtVsCodes}
  Let $\Ext\colon \F_q^n \times [D] \to \F_q^m$ be a linear function
  in the first argument. For each $z \in [D]$, let
  $G_z \in \F_q^{m \times n}$ be such that $\Ext(x, z) = G_z \cdot x$
  for $x \in \F_q^n$.  Then, $\Ext$ is a
  $((1-\delta) n, \eps)$-extractor for symbol-fixing sources (entropy
  measured in $q$-ary symbols) if any only if $\{G_z\}_{z \in [D]}$ is
  an $[n,\delta,\eps]_q$-erasure code family.
\end{lem}

The following duality between linear affine extractors and lossless
condensers was also demonstrated in \cite{Che09}:

\begin{lem}\cite{Che09}  \label{lem:duality}
  Let $G\in \F_q^{m\times n}$ and $H\in \F_q^{(n-m)\times n}$ be
  matrices of full row rank such that $G H^\top = 0$.  Define
  $g\colon \F_q^n \to \F_q^m$ by $g(x) = G\cdot x$ and
  $h\colon \F_q^n \to \F_q^{n-m}$ by $h(x) = H\cdot x$.  Then, for any
  affine space $A = a+V \subseteq \F_q^n$ (where $V$ is a vector
  subspace and $a \in \F_q^n$ is a translation) and a dual affine
  space $B = b+V^\perp \subseteq \F_q^n$ (where $V^\perp$ is the dual
  of $V$ and $b \in \F_q^n$), $g$ is an extractor (with zero error)
  for the affine source uniformly distributed on $A$ if and only if
  $h$ is a lossless condenser for the uniform distribution on $B$.
\end{lem}

This, in particular, implies an equivalence between linear seeded
affine (in particular, symbol-fixing) extractors and lossless
condensers. One can be constructed from the other by applying the
above duality to the linear function defined by each individual seed.
An interesting corollary of this is that, unlike general seeded
extractors and lossless condensers (cf.\ \cites{RTS00,CRVW02}), the
optimal seed lengths for linear seeded affine (or symbol-fixing)
extractors and lossless condensers must be equal.

Using state-of-the-art constructions of linear extractors and lossless
condensers (for general sources), \cite{Che09} construct erasure code
families of polynomial and quasi-polynomial size. In particular, the
following is a consequence of using a linear instantiation of the
so-called GUV condenser \cite{GUV09} and (an improvement of)
Trevisan's extractor \cite{RRV99}.

\begin{thm} \label{thm:Che} \cite{Che09}\footnote{We remark that
    \cite{Che09} does not explicitly use the language of our
    \cref{def:erasure:family}; however, the result can be recast in
    this way.}  There are explicit constructions of
  $[n,\delta,\eps]_q$-erasure families achieving rates at least
  $1-\delta-\eta$ and size $\poly(n^{1/\eta}/\eps)$ or
  $\exp(O((\log^2 n)\log(1/\eta)\log(1/\eps)))$.  \qedhere \qed
\end{thm}

Combined with \cref{lem:ExtVsCodes,lem:duality}, we note that
\cref{lem:existence:code} implies the following consequence on the
parameters achieved by strong, seeded, linear symbol-fixing extractors
and lossless condensers:

\begin{coro} \label{coro:ext:cond} For any $\delta \in (0,1)$, there
  are functions
  $\Ext\colon \F_q^n \times \zo^d \to \F_q^{(\delta-\eta) n}$ and
  $\Cond\colon \F_q^n \times \zo^d \to \F_q^{(\delta+\eta) n}$ such
  that
  \begin{enumerate*}[label=(\arabic*)]
  \item For each $z \in \zo^d$, $\Ext(\cdot, z)$ and $\Cond(\cdot, z)$
    are $\F_q$-linear functions; and
  \item The functions $\Ext$ and $\Cond$ are a strong symbol-fixing
    extractor and lossless condenser, respectively for input ($q$-ary)
    entropy $\delta n$ and error $\eps$, where
    $d=\log(1/\eta \eps)+O(1)$.
  \end{enumerate*} \qedhere \qed
\end{coro}

\section{Randomness-Efficient Linear Erasure Codes}
\label{sec:codes}

In this section, we present and analyze our construction of a
constant-sized $[N, \delta, \eps]_q$-erasure code family achieving the
optimal rate of $1-\delta-\eta$ for any $\eta>0$. Namely, we prove
\cref{thm:codes:simplified}.

\subsection{The Construction}
\label{sec:code:constr}

\begin{figure}[!t]
  \begin{center}
    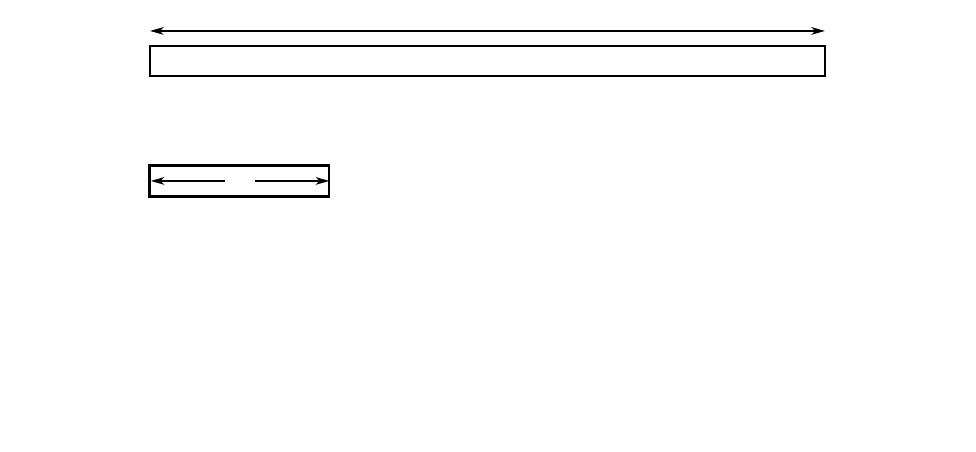
  \end{center}

  \caption{Construction of the erasure code family in \cref{sec:codes}
    from the decoder's perspective (codeword at the top, decoding at
    the bottom).  The function $\Ext\colon [N] \times [D] \to [M]$ is
    a strong $(\log N - \Delta, \nu)$-extractor for
    $\Delta=-\log(1-\delta)$ and $\nu = O(\eps \eta^2)$.  The inner
    code family $\ensIn$ is an $[L, \delta+2\eta, \mu]_q$-erasure code
    family for $\mu = O(\eps \eta)$. The construction contains a code
    for each choice of $(z,\Ci) \in [D] \times \ensIn$.  The extractor
    assigns codeword positions to outer code blocks, in order.
    Occasionally, this causes overfull blocks, in which case the
    corresponding codeword position is frozen to zero (as depicted).
  }
  \label{fig:main}
\end{figure}

Our main construction implements the following procedure that is
depicted in \cref{fig:main}. We depict the construction with the
various parameters involved left uninstantiated, and defer the
concrete balancing of the parameters to the later sections.

A foundational object used in our construction is an outer code that
can approximately achieve the Singleton bound over a constant-sized
alphabet.  Since we need multiple instantiations of such outer codes
achieving different trade-offs, we formulate an abstract definition
below.

\begin{defn} \label{defn:outer} For fixed constants
  $\alpha, \gamma \geq 1$ and $\beta \geq 0$, we say that outer codes
  are (strongly) $(\alpha, \beta, \gamma)$-attainable if for any
  $\eta\in(0,1)$, there is a $Q_0 \leq 2^{O((1/\eta)^\beta)}$ and
  $N_0 = O(1/\eta^\gamma)$ such that for any $Q \geq Q_0$ that is a
  power of $q$, the following holds: There is a (strongly) explicit
  construction of an $\F_q$-linear code $\Co \subseteq \F_Q^N$ of any
  length $N \geq N_0$, relative distance greater than $\eta$, and rate
  at least $1-O(\eta^{1/\alpha})$.  Moreover, $\Co$ can be encoded and
  erasure-decoded up to $\eta N$ erasures in $\tilde{O}(N \log Q)$
  time\footnote{Using trivial padding, it suffices to construct codes
    for a sufficiently dense infinite set of lengths $N \geq N_0$
    (i.e., as long as for each available length $N$, the next smallest
    length in the family is at most $N(1+o(1))$). Moreover, for any
    available construction over $\F_Q$, it is straightforward to
    increase the alphabet to any larger $q$ power $Q' > Q$ by
    interpreting the available codes as codes over the base field
    $\F_q$ and re-bundling the symbols to any desired packet length.
  }.
\end{defn}

We record our starting point for outer codes below.  Later, we
bootstrap our results by using all the machinery that we shall develop
to tighten these parameters (cf.\ \cref{coro:alpha}).

\begin{prop} \label{prop:outer} For any fixed $\eps>0$, outer codes
  are strongly $(3+\eps, 0, 1)$-attainable.
\end{prop}

\begin{proof}
  We take $Q_0 = q$, so that $\beta = 0$.  A natural idea is to use
  expander-based codes such as \cites{SS96,Spi95,AEL95} that are
  equipped with linear time encoders and decoders.  However, we are
  unable to verify whether these constructions are strongly explicit
  (even if the underlying expander graph construction is strongly
  explicit) due to the layered nature of the constructions in
  systematic form\footnote{In general, expander codes are more
    naturally defined in terms of a parity check matrix, but
    systematic representation is needed to avoid the need for costly
    Gaussian elimination to transition from parity checks to a
    generator matrix.}.  To avoid this, we use the classical
  concatenated codes of \cite{Jus72} that concatenate arbitrary
  Reed-Solomon codes with a family of inner codes of logarithmic
  length, most of which are on the Gilbert-Varshamov bound (which, for
  any desired relative distance $\delta_0 > 0$, implies a rate of
  $1-O(\delta_0 \log_q(1/\delta_0))$).  This provides a concatenated
  code of relative distance larger than $\eta$ and rate
  $1-O(\sqrt{\eta} \log_q(1/\eta))$ (i.e., $\alpha=2+\eps$ is
  attained).  However, if an exponentially large inner code family (in
  the inner code block length) is used, as \cite{Jus72} does, the
  minimum block length $N_0$ for the final code would not be
  polynomially bounded in $1/\eta$. To avoid this, Forney's code
  concatenation \cite{For66} combined with an exhaustive search for
  the inner code \cite{PR11} can be used. The side effect of this,
  though, is that the exhaustive search takes exponential time in the
  block length of the inner code, which in this case is logarithmic in
  the block length $N$ of the final code. That is, the construction
  would take $\poly(N)$ time, not achieving strong explicitness. To
  address this\footnote{%
    Alternatively, an expander-based construction based on
    \cite{Spi95} could be used as the inner code.  Since the inner
    codeword lengths are only logarithmic in the final block length,
    the resulting code construction is strongly explicit as long as
    the inner code is explicit (not necessarily strongly
    explicit). Doing so would also ensure that the dependence of the
    runtime of the construction on $1/\eta$ is polynomial, if such a
    dependence is required.  }, a two-layered code concatenation can
  be used (i.e., Forney's concatenated code construction used as its
  own inner code), at the cost of increasing $\alpha$ to
  $3+\eps$. Concretely, letting $\eta_0=\sqrt[3]{\eta}$, we can
  concatenate a Reed-Solomon outer code of relative distance $\eta_0$
  and rate $R_1\ge1-\eta_0$, a short ``intermediate'' Reed-Solomon
  code of relative distance $\eta_0$ and rate $R_2\ge1-\eta_0$ and,
  finally, a $q$-ary linear code on the Gilbert-Varshamov
  bound\footnote{Here we use the estimate
    $H_q(x)=\Theta_q(x\log{1/x})$ on the $q$-ary entropy function
    $H_q$ defining the bound (cf.\
    \cite{codingbook}*{Proposition~3.3.8}).}  of relative distance
  $\eta_0$ and rate $R_3\ge1-O_q(\eta_0\log{1/\eta_0})$ that is found
  by an exhaustive search (see \cite{PR11}*{Theorem~2}). Using this
  two-layered code concatenation, we can get an $\F_q$-linear code of
  relative distance $\eta^3_0=\eta$ and rate
  $R_1R_2R_3\ge 1-O_q(\sqrt[3]{\eta}\log{1/\eta})$. This results in
  $\alpha=3+\eps$. To implement this concatenation, it suffices to set
  the block length of each of the three codes to be
  $\Omega_q(1/\eta_0)$.  Therefore, the construction is valid for any
  total block length $N\ge N_0$ for some $N_0=O(1/\eta)$, leading to
  the conclusion that $\gamma=1$. Moreover, this inner code has block
  length $O(\log{\log{N}})$, so the construction time is still
  quasi-linear.  The time bounds on encoding and erasure decoding
  follow by the standard FFT-based polynomial evaluation and
  interpolation and a naive Gaussian for the erasure decoding of inner
  code blocks. The strong explicitness is implied by the strong
  explicitness of Reed-Solomon codes.
\end{proof}

In the sequel, we assume that outer codes are strongly
$(\alpha, \beta, \gamma)$-attainable. Accordingly, let
$\Co \subseteq \F_Q^{M}$, where $Q = q^{\ell}$, be an $\F_q$-linear
outer code that achieves a minimum distance greater than $\eta M$ at
rate $\Ro$.  We pick a suitable value for $\ell$ in the analysis.
From \cref{defn:outer}, we can take $\Ro \geq 1-O(\eta^{1/\alpha})$
and, in this regard, need to ensure that $\ell \geq \ell_0$, for some
$\ell_0 = O((1/\eta)^\beta / \log q)$ and that
$M = \Omega(1/\eta^\gamma)$.

Let $\Ext\colon [N] \times [D] \to [M]$ (what we call the ``shuffler
extractor'') be a strong $(\log N - \Delta, \nu)$-extractor for
$\Delta \coloneq -\log(1-\delta)$ and an appropriate $\nu$ that shall
be determined in the analysis.  For a given seed $z \in [D]$, we use
the shorthand $\Ext_z(x)$ for $\Ext(x,z)$.

For each $i \in [M]$, let $S ^z_i \coloneq \Ext_z^{-1}(i)$; i.e., the
set of inputs that the extractor maps to $i$ given seed $z$.  Below,
we observe that these sets generally intersect $S$ evenly.

\begin{prop} \label{prop:balance:general} Let $\eps_1,\eps_2,\eps_3$
  be such that $\eps_1\eps_2\eps_3\ge 2\nu$.  For any set
  $S\subseteq[N]$ with $|S|\geq N/2^{\Delta}$, the following holds.
  For all but at most an $\eps_1$ fraction of the choices of the seed
  $z$, all but at most an $\eps_2$ fraction of the choices of $i$
  satisfy $|S^z_i\cap S| \in (1 \pm \eps_3) |S|/M$.
\end{prop}

\begin{proof}
  We use a standard averaging argument. Let $U_S$ denote the uniform
  distribution on $S$.  Since $H_{\infty}(U_{S})\ge \log{N}-\Delta$,
  and $\Ext$ is a strong $(\log{N}-\Delta,\nu)$ extractor, the
  definition of extractors implies that for $Z \sim [D]$,
  $(Z,\Ext(U_S, Z)) \sim_\nu U_{[D]\times [M]}$.  Note that
  $\Pr[\Ext(U_S, z)] = i] = |S_i^z \cap S|/|S|$, implying that
  \begin{align} \label{eqn:ext:def}
    \sum_{z\in[D]}\sum_{i\in[M]}\left|\frac{|S^z_i\cap
        S|}{D|S|}-\frac{1}{DM}\right| &\leq 2\nu.
  \end{align}
  For $z\in[D]$, let
  $T_z \coloneq \E_{i\sim {[M]}} \left[\left||S^z_i\cap
      S|-|S|/M\right|\right]$ so that \eqref{eqn:ext:def} becomes
  $\E[T_Z] \leq 2 \nu |S|/M \leq \eps_1 \eps_2 \eps_3 |S|/M$. By
  Markov's inequality\footnote{Namely, for any non-negative random
    variable $X$ and $a>0$, we have $\Pr[X\geq a]\leq \E[X]/a$.}
  applied to the random variable $T_Z$, it follows that for all but at
  most an $\eps_1$ fraction of $z\in[D]$, we have
  $T_z\leq \eps_2\eps_3 |S|/M$. For any such $z$, we can apply
  Markov's inequality again on the expression that defines $T_z$ to
  conclude that for all but at most an $\eps_2$ fraction of the
  choices of $i \in [M]$, we have
  $\left||S^z_i\cap S|-|S|/M\right| \leq \eps_3 |S|/M$.  The claim
  follows.
\end{proof}

Note that in particular, by setting $S=[N]$ in
\cref{prop:balance:general}, we deduce that the sets $S_i^z$ are
generally balanced in size; namely, that we have the following.

\begin{prop} \label{prop:balance} Let $\eps_1,\eps_2,\eps_3$ be such
  that $\eps_1\eps_2\eps_3 \geq 2\nu$. For all but at most an $\eps_1$
  fraction of the choices of seed $z$, we have that all but at most an
  $\eps_2$ fraction of the choices of $i$ satisfy
  $|S^z_i| \in (1 \pm \eps_3) N/M$. \qedhere \qed
\end{prop}

Our ensemble of codes contains a collection $\ensIn^z$ of codes for
each fixed choice of $z\in[D]$. Each collection $\ensIn^z$ of codes
corresponds to the codes of an inner ensemble $\ensIn$ of linear
$q$-ary codes of dimension $\ell$. Concretely, given the inner code
ensemble $\ensIn$, our final erasure code family is
$\ens \coloneq \bigcup_{z\in[D]}\ensIn^z$, where $|\ensIn^z|=|\ensIn|$
for any seed $z\in[D]$. Next, we describe how to choose $\ensIn$ and
define the code collections $\ensIn^z$.

We assume that $\ensIn$ in turn is an
$[\ell/\Ri,\delta+2\eta,\mu]_q$-erasure code family, for an
appropriate parameter $\mu$ to be determined in the analysis, and
achieves rate $\Ri \geq 1-\delta-O(\eta)$. This makes the size of the
final ensemble of codes equal to $|\ens|=|\ensIn|D$.  As long as the
dimension $\ell$ is a constant or slightly super-constant (e.g.,
$\ell = O(\sqrt{\log N})$), the ensemble $\ensIn$ whose existence is
guaranteed by \cref{lem:existence:code} can be constructed explicitly
by a trivial exhaustive search.

For a given seed $z$ and $\Ci \in \ensIn$, we define a code
$\Ci^z\subseteq\F^N_q$. Our code construction is a usual concatenated
code followed by a ``shuffler'' layer that we now explain.  First, a
codeword of $\Co$, denoted by $c = (c_1, \ldots, c_M)$ is constructed
from the message. Recall that each $c_i$ is a $q$-ary vector of length
$\ell$. Then, each $c_i$ is further encoded to a codeword of $\Ci$
which is a $q$-ary vector of length $L \coloneq \ell/\Ri$. Let
$c'_i \in \F_q^L$ denote the resulting encoding of $c_i$.

We set the parameters so that $N=LM$.  The final codeword
$C \in \Ci^z\subseteq \F_q^N$ is constructed as follows.  Recall the
notation $S^z_i \coloneq \Ext_z^{-1}(i)$.  These sets (for the fixed
choice of $z$) are expected to partition $[N]$ nearly uniformly by
\cref{prop:balance} (for parameters to be specified).  For each
$i \in [M]$, the coordinate positions of $C$ that lie in $S^z_i$
collect the $q$-ary symbols of $c'_i$.  This is done with respect to
an arbitrarily fixed ordering, such as the natural integer ordering of
the coordinate indices. Any leftover symbols in $C$ that remain
unassigned, due to some $S^z_i$ being larger than $L$, are frozen to
zeros.  On the other hand, in case $|S^z_i| < L$, any leftover symbols
of $c'_i$ are not be included in the final codeword and are discarded.

For any $\Ci \in\ensIn$, we use $\Ci^z$ to denote the code constructed
from $\Ci$ and $\Ext_z$, and define $\ensIn^z$ as the collection of
codes $\{\Ci^z\colon \Ci\in\ensIn\}$. We remind that our final erasure
code family is $\ens=\bigcup_{z\in[D]}\ensIn^z$.

It is immediate to observe that the resulting final code is linear.
The erasure correction properties of the code ensemble are analyzed
below.

\begin{lem} \label{lem:code:analysis} The code ensemble defined in
  this section (containing a code for each element of
  $[D] \times \ensIn$) is an $[N,\delta,\eps]_q$-erasure family for
  some choices of the parameters $\nu = O(\eps \eta^2)$ and
  $\mu = O(\eps \eta)$, achieving rate at least
  $1-\delta-O(\eta^{1/\alpha})$.
\end{lem}

\begin{proof}
  The rate of each code in the ensemble is readily seen to be
  $R = \Ri \Ro $, assuming that the code can recover the outer
  codeword when there are no erasures (which, in turn, follows as a
  special case of the erasure correction analysis that we show
  below). Using the fact that $\alpha \geq 1$, we get
  $R \geq (1-\delta-O(\eta))(1-O(\eta^{1/\alpha})) \geq
  1-\delta-O(\eta^{1/\alpha})$.

  Let us now consider any pattern of up to $\delta$ fraction of
  erasures and denote by $S \subseteq [N]$ the set of non-erased
  positions.  We have that $|S| \geq (1-\delta)N$.  Our goal is to
  show that all but an $O(\eps)$ fraction of the codes in the ensemble
  constructed in this section can correct\footnote{It is important to
    note that the choice of $S$ does not depend on which code in the
    ensemble is being picked. This is a fundamental aspect of the
    model, since otherwise the problem would reduce to the standard
    rate-distance trade-off of $q$-ary codes in the Hamming metric and
    having an ensemble would not make a difference.}  the erasure
  pattern corresponding to $S$.

  First, let us invoke \cref{prop:balance} for
  $(\eps_1,\eps_2,\eps_3) = (\eps/3,\eta/4,\eta)$, which requires
  $\nu = O(\eps \eta^2)$, and assume in the sequel that a seed
  $z \in [D]$ is picked so that the conclusion of the proposition
  holds. By doing so, we discard up to an $\eps/3$ fraction of the
  codes in the ensemble.  For the given $z$, we know that all but at
  most an $\eps_2 = \eta/4$ fraction of the inner code blocks
  $i \in [M]$ satisfy
  $|S^z_i| \in (1 \pm \eps_3) N/M = (1 \pm \eta) L$. Call $i$
  \emph{non-deficient} if this property holds for $S^z_i$ and
  \emph{deficient} otherwise.  Therefore, for any inner code
  $\Ci \in \ensIn$, given a codeword $C \in \Ci^z \subseteq \F_q^N$,
  for all non-deficient blocks, the number of positions that are
  frozen to zeros is at most $(1+\eps_3) N/M - L = \eta L$.

  Similarly, invoke \cref{prop:balance:general} for
  $(\eps_1,\eps_2,\eps_3) = (\eps/3,\eta/4,\eta)$ and the given choice
  of $S$.  Discard all choices of seed $z$ that are excluded by this
  result. By now, we have discarded a $2\eps/3$ fraction of the codes
  in the ensemble. Assume that $z$ survives this exclusion as
  well. Call $i \in [M]$ \emph{balanced} if
  $|S^z_i\cap S| \geq (1 - \eps_3) |S|/M \geq (1-\eta)(1-\delta)L \geq
  (1-\delta-\eta)L$, where the second inequality holds by the
  assumption $|S| \geq (1-\delta)N$ and the choice of the length
  parameter $N$.

  Altogether, we have ensured that for any inner code
  $\Ci \in \ensIn$, given a codeword $C \in \Ci^z \subseteq \F_q^N$,
  all but at most an $\eta/2$ fraction of the blocks $i \in [M]$ are
  balanced and non-deficient.  Moreover, for any such block, the
  number of non-erased positions that are not frozen to zeros is at
  least $(1-\delta-2\eta)L$.

  Denote by $c=(c_1, \ldots, c_M) \in \Co \subseteq \F_Q^M$ the outer
  codeword from which $C$ is obtained, which we wish to recover given
  the erasures.  Recall that since $\Co$ can be recovered from any
  $\eta$ fraction of erasures, it suffices to recover at least some
  $1-\eta$ fraction of the symbols in $c$. Our task is to analyze the
  proportion of the choices of $\Ci \in \ensIn$ for which this is
  possible.

  Let $G \subseteq [M]$ be the set of all balanced and non-deficient
  blocks. We know that $|G|\ge(1-\eta/2)M$.  The erasure correction
  properties of $\ensIn$ ensure that for any $i \in G$ and all but a
  $\mu$ fraction of the choices of $\Ci \in \ensIn$, the code $\Ci^z$
  allows for the recovery of the outer code symbol $c_i \in \F_Q$.
  Equivalently, denoting by $A(i, \Ci) \in \{0,1\}$ the indicator for
  the event that $c_i$ cannot be recovered by $\Ci^z$ when we
  uniformly and independently sample $i\in G$ and $\Ci\in{\ensIn}$, we
  can write this probability as
  \[
    \underset{\Ci \sim {\ensIn}}{\E} \underset{i \sim G}{\E} [A(i,
    \Ci)] \leq \mu \Rightarrow \underset{\Ci \sim {\ensIn}}{\Pr}
    \left[ \underset{i \sim G}{\E} [A(i, \Ci)] > \eta/2\right] \leq
    2\mu/\eta\leq \eps/3,
  \]
  where we have used Markov's bound for the last inequality.  As long
  as $\mu \leq \eps \eta/6$, which we ensure to be the case, the right
  hand side is at most $\eps/3$.  So far, we have discarded an at most
  $2\eps/3$ fraction of the seeds $z\in[D]$ and the corresponding code
  ensembles $\ensIn^z$.  For any remaining seed $z\in[D]$, we know
  that there is at most an $\eps/3$ fraction of the codes
  $\Ci\in\ensIn$ that do not satisfy the above condition, and we
  discard such codes.  Overall, we have discarded at most an $\eps$
  fraction of the codes in $\ens$.

  It suffices to show that we can recover at least a $1-\eta$ fraction
  of the symbols of $c$ for any remaining code. We can guarantee that
  for all remaining codes, all but at most an $\eta/2$ fraction of the
  symbols corresponding to blocks in $G$ can be recovered and
  $|G|\ge(1-\eta/2)M$.  Therefore, for any remaining code, at least a
  $(1-\eta/2)|G|\ge(1-\eta/2)^2\ge1-\eta$ fraction of the symbols in
  $c = (c_1, \ldots, c_M)$ can be recovered.  Finally, the outer code
  ensures a full recovery of the codeword.
\end{proof}

\subsection{Setting Up the Parameters} \label{sec:code:params}

For the concrete choice of the shuffler extractor $\Ext$, we assume
that the extractor is implied under the assumption that extractors are
strongly $(\gamma_1,\gamma_2)$-attainable (cf.\
\cref{defn:ext:constr}).  In particular, by \cref{prop:ext:constr} one
can pick $\gamma_1 = 4$ and $\gamma_2 = 2$.  For our application, we
can set $N = 2^n$, $D = 2^d$, and $M=2^m$.

Using the above shuffler extractor in our construction leads to our
final explicit construction, which is summarized below.  We note that
if the underlying extractor is only attainable and not strongly
attainable, our code construction would still be explicit albeit not
strongly explicit.

\begin{thm} \label{thm:explicit} Assume that outer codes are
  (strongly) $(\alpha,\beta,\gamma)$-attainable and that extractors
  are (strongly) $(\gamma_1, \gamma_2)$-attainable.  Fix any
  $\delta \in [0,1)$.  For parameters $\eps>0$ and $\eta > 0$, there
  is an
  $N_0 = O(((\eps \eta^2)^{-\gamma_2}+\eta^{-\beta})/\eta^{\gamma})
  =(1/\eps\eta)^{O(1)}$ such that the following holds: For all
  $N \geq N_0$, there is a (strongly) explicit construction of an
  $[N,\delta,\eps]_q$-erasure code family that achieves rate at least
  $1-\delta-\eta^{1/\alpha}$.  Moreover, the size of the code family
  is $O(1/((\eps \eta^{2})^{1+\gamma_1} \log q))=1/(\eps \eta)^{O(1)}$
  which can be taken to be a power of two.  Furthermore, after a
  one-time pre-processing time of
  $q^{O(((\eps \eta^2)^{-\gamma_2}+
    \eta^{-\beta})^2\eps^{-1}\eta^{-2})}=\exp((\eps\eta)^{-O(1)})$,
  each code in the family can be encoded and erasure decoded (whenever
  possible) in quasi-linear time.
\end{thm}

\begin{proof}
  For the inner code ensemble $\ensIn$, we use the result of
  \cref{lem:existence:code} combined with the parameters required by
  the statement of \cref{lem:code:analysis} that, by setting
  $\mu = O(\eps \eta)$, yields
  $|\ensIn| = O(1/(\eta \mu \log q)) = O(1/(\eps \eta^2 \log q))$.
  The number of seeds for the shuffler extractor is, from
  \cref{defn:ext:constr}, $D = O(1/\nu^{\gamma_1})$, noting that the
  entropy deficiency $\Delta$ is a constant.  The result then follows
  by applying \Cref{lem:code:analysis}, recalling that
  $\nu = O(\eps \eta^2)$.  Note that this only provides a construction
  for infinitely many choices of the block length $N$; however, this
  can be corrected by trivial padding and a slight adjustment of the
  parameters that does not affect the asymptotics.

  The length $L=N/M$ of each code in $\ensIn$ corresponds to the
  entropy loss $n-m$ of the shuffler extractor (including the source
  entropy deficiency).  It can be adjusted to a desired value and must
  be picked as small as possible (to optimize the time needed for the
  exhaustive search) but subject to the following considerations.
  
  \begin{enumerate}[label=(\roman*)]
  \item The entropy loss of the shuffler extractor (including the
    entropy deficiency of the original source) which in the language
    of \cref{defn:ext:constr} is bounded by
    $\Delta+\gamma_2 \log(1/\nu)+O(1)= \log(1/(\eps
    \eta^2)^{\gamma_2})+O(1)$.
  \item The minimum alphabet size $Q_0$ of the outer code, which is
    $2^{O(1/\eta^\beta)}$.
  \item The minimum length allowed by the existence result of erasure
    code families (\cref{lem:existence:code}), which is
    $O(\log(1/\mu)/\mu)=\tilde{O}(1/(\eps \eta))$.
  \end{enumerate}

  A value of
  $L = O(1/\nu^{\gamma_2} + 1/\eta^\beta)= O((\eps
  \eta^2)^{-\gamma_2}+\eta^{-\beta})$ is compatible with all the above
  requirements (considering the fact that $\gamma_2 \geq 2$).  This,
  combined with the minimum block length of the outer code, also
  determines the minimum block length of the final code ensemble,
  which becomes $O(L/\eta^\gamma)$.  The guarantee on the size of the
  ensemble being a power of two can be achieved by ensuring that
  $\ensIn$ is a power of two (combined with the seed of the extractor
  in \cref{thm:RVW} being a bit string).

  The one-time pre-processing procedure involves an exhaustive search
  for the inner code ensemble.  This would take an amount of time
  upper bounded by
  $q^{L^2 |\ensIn|} \cdot 2^L \cdot \poly(L) = q^{O(((\eps
    \eta^2)^{-\gamma_2}+\eta^{-\beta})^2\eps^{-1}\eta^{-2})}$,
  enumerating all possible linear code ensembles of a given size and
  then checking for all erasure patterns.  Since the construction is
  based on code concatenation, strong explicitness guarantee holds as
  long as the outer code construction is strongly explicit and that
  the shuffler extractor is computable in polynomial time in its input
  length.  Finally, considering that the outer code is encodable and
  decodable in quasi-linear time leads to a quasi-linear time encoder
  and decoder for the overall code.
\end{proof}

\begin{remark}[Avoiding Exhaustive Search] \label{rem:exhaustive}
  Instead of an exhaustive search for the inner code that was done in
  the proof of \cref{thm:explicit}, it is also possible to use
  explicit ensembles such as those constructed in \cite{Che09} (i.e.,
  \cref{thm:Che}).  These ensembles achieve a polynomial size in the
  block length of the inner code ensemble (which is a constant
  polynomially depending on $\eta$ and $\eps$) and polynomial size in
  the error parameter $\eps$, albeit exponential size in the gap to
  capacity parameter $\eta$ for the specific instantiations recorded
  in \cref{thm:Che}.
\end{remark}

With foresight, \cref{coro:alpha} (that in turn, uses
\cref{thm:explicit} with the value of $\alpha \approx 3$ provided by
\cref{prop:outer}) shows that we can pick the choice $\alpha = 1$, for
some absolute constants $\beta$ and $\gamma$. Using \cref{coro:alpha},
and by also picking $\gamma_1 = 4$ and $\gamma_2 = 2$ according to
\cref{prop:ext:constr}, we can rewrite a simplified version of
\cref{thm:explicit} that appears below.

\begin{coro} \label{coro:explicit} Fix any $\delta \in [0,1)$.  For
  parameters $\eps>0$ and $\eta > 0$, there is an
  $N_0 = \poly(1/(\eps \eta))$ such that the following holds: For all
  $N \geq N_0$, there is a strongly explicit construction of an
  $[N,\delta,\eps]_q$-erasure code family that achieves rate at least
  $1-\delta-\eta$.  Moreover, the size of the code family is
  $O(1/(\eps^5 \eta^{10} \log q))$ which can be taken to be a power of
  two.  Furthermore, after a one-time pre-processing time of
  $q^{\poly(1/(\eps \eta))}$, each code in the family can be encoded
  and erasure decoded (whenever possible) in quasi-linear
  time. \qedhere \qed
\end{coro}

Combined with \cref{lem:ExtVsCodes,lem:duality} \cref{coro:explicit}
immediately translates into an explicit construction of seeded linear
symbol-fixing extractors and lossless condensers that achieve a
constant seed length (only depending on normalized entropy loss and
error):

\begin{coro} \label{coro:explicit:bitfixing} Fix any
  $\delta \in (0,1]$. For parameters $\eps>0$ and $\eta>0$, there is
  an $N_0 = \poly(1/(\eps \eta))$ such that the following holds: For
  all $N \geq N_0$, there are explicit constructions of functions
  $\Ext\colon \F_q^N \times \zo^d \to \F_q^{(\delta-\eta) N}$ and
  $\Cond\colon \F_q^N \times \zo^d \to \F_q^{(\delta+\eta) N}$ where
  $d=5\log(1/(\eps \eta^2 \log q))+O(1)$.  Moreover,
  \begin{enumerate*}[label=(\arabic*)]
  \item For each $z \in \zo^d$, $\Ext(\cdot, z)$ and $\Cond(\cdot, z)$
    are $\F_q$-linear functions; and
  \item The functions $\Ext$ and $\Cond$ are a strong
    $(\delta N, \eps)$-extractor and a
    $(\leq \delta N, \eps)$-lossless condenser, respectively, for
    symbol-fixing sources.
  \end{enumerate*}
  \qedhere \qed
\end{coro}

\section{Optimal Codes on Bipartite Graphs}
\label{sec:graphs}

In this section, $M$ and $N$ are sufficiently large integers, and we
assume for technical reasons that $M$ is a power of two.  As before,
we assume that outer codes are $(\alpha, \beta, \gamma)$-attainable,
and that strong extractors are $(\gamma_1,\gamma_2)$-attainable.  The
goal is to explictly construct an $[M,N,\drow,\dcol]_q$-graph code,
according to \cref{def:bigraph:code}, at rate approaching the optimal
$(1-\drow)(1-\dcol)$ arbitrarily closely.  Of particular interest is
when $M=N$ and $\drow = \dcol$, but we allow a more general choice of
parameters. In fact, we use the unbalanced case in \cref{sec:AEL}.
The cost paid for the gap to the optimal rate is on how large $M$ and
$N$ are required to be, as well as the (explicit) code construction
time and (quasi-linear) erasure correction time, and this can be
optimized to achieve a sub-constant gap to capacity as well.

\newcommand{\myarrow}{%
  \mathrel{%
    \tikz[baseline=-0.5ex]{%
      \draw[->, line width=1pt, >={Stealth[length=0.09in,
        width=0.05in]}] (0,0) -- (1,0); }%
  } } \newcommand{\figFamily}{
  \begin{array}{lcr}
    \myarrow &\cC_1 \in \ens &\myarrow \\
    \myarrow &\cC_2 \in \ens &\myarrow \\ 
             &\vdots \\
    \myarrow & \cC_{M_0} \in \ens &\myarrow \\ 
    \myarrow &\cC_1 \in \ens &\myarrow \\
    \myarrow &\cC_2 \in \ens &\myarrow \\ 
             &\vdots \\
    \myarrow & \cC_{M_0} \in \ens &\myarrow \\ 
             &\vdots \\
  \end{array}}
  
\begin{figure}[!t]
  \begin{center}
    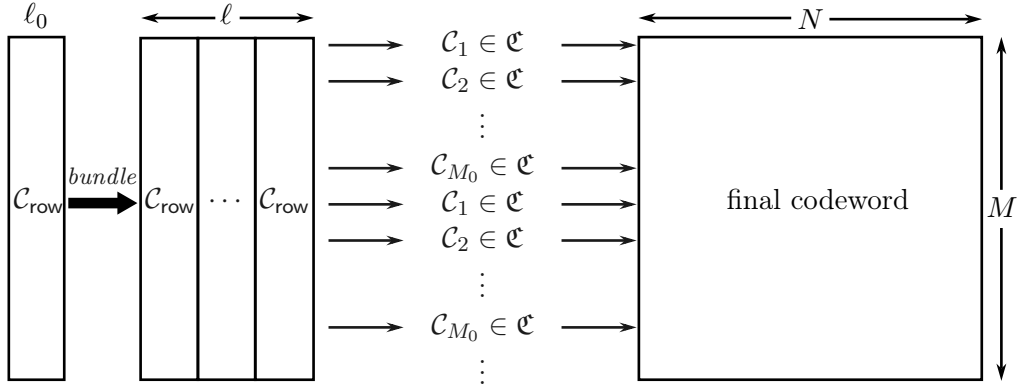
  \end{center}

  \caption{Construction of the bipartite graph codes in
    \cref{sec:graphs}.  The row-erasure correction code $\Crow$ (of
    alphabet size $q^{\ell_0}$) is bundled to provide a sufficient
    number $\ell$ of columns.  Then, each row of the matrix consisting
    of codewords of $\Crow$ is encoded by a codeword from the erasure
    code family $\ens$ to provide column-erasure correction.  }
  \label{fig:graph:main}
\end{figure}

\subsection{The Construction} \label{sec:graph:constr} Our starting
point is an $\F_q$-linear code $\Crow \subseteq \F_{q^\ell}^M$ of rate
$\Rrow$ that can correct any erasure pattern as long as at least
$(1-\drow)(1-\eta)$ fraction of the symbols remain.  The naming
$\Crow$ is chosen to remind that the code is responsible for
correcting \emph{row erasures}.  The parameters $\ell$ and $\eta$ are
to be determined later.  Observe that if a code over
$\F_{q^{\ell_0}}^M$ with the above guarantees is available for some
$\ell_0 | \ell$, it is possible to artificially increase the alphabet
size to the desired $q^\ell$ by simply bundling $\ell/\ell_0$
independent codewords as an element of $\F_{q^\ell}^M$ without
affecting the rate or distance.

Our eventual choice of $\ell$ turns out to be significantly large
(e.g., linear in $M$ when $\drow > 0$ and even much higher when
$\drow = 0$) and that allows the use of a Reed-Solomon or any MDS code
for $\Crow$ and achieving rate $\Rrow = (1-\drow)(1-\eta)$. However,
since Reed-Solomon codes require polynomial-sized alphabets (in
length), doing so would cause the slight inconvenience of affecting
the minimum possible value for $N$ (which we ideally wish to only
depend on the constant parameter $\eta$ that determines the gap to the
optimal rate).  To address this, we distinguish two cases:
\begin{description}
\item[Case~1, where $\drow = 0$.] In this case, the number of rows $M$
  can, without loss of generality, be thought of as a constant $M_0$
  (only depending on the gap to optimal rate). This is because there
  are no row erasures, and as long as an $[M_0,N,0,\dcol]_q$-graph
  code $\mathcal{C}\subseteq \F^{M_0\times N}_q$ with some constant
  $M_0$ is constructed, the number of rows can be extended to any
  multiple $M$ of $M_0$ by stacking independent $M/M_0$ codewords on
  top of each other to achieve the desired number of rows without
  affecting the rate. Therefore, in this case the use of a
  Reed-Solomon code over a constant-sized alphabet would not cause an
  undesirable side effect. Namely, we pick a Reed-Solomon code
  $\Crow \subseteq \F_{q^{\ell_0}}^M$ for $q^{\ell_0} = O(M)$, and
  then extend the alphabet size to $q^\ell$ by the bundling procedure
  described above. By the above discussions, in this case we can
  assume $M=M_0$ for some minimum constant number of rows. The exact
  value of $M_0$ is to be determined by other components of our
  construction below.  Importantly, this special case is the subject
  of study in \cref{sec:AEL}.

\item[Case~2, where $\drow > 0$.] In this case, we appeal to the
  result of \cref{thm:almost:MDS} (that, in turn, only relies on the
  results for the special case $\drow=0$ in this section\footnote{We
    distinguish the case $\drow=0$ not only to optimize the
    parameters, but also to avoid a circular argument.}). Namely, we
  pick $\Crow \subseteq \F_{q^{\ell_0}}^M$ for
  $q^{\ell_0} = 2^{O(1/\eta^{3(1+\gamma_1)})}$ and achieving rate at
  least $(1-\drow)(1-O(\eta^{1/\alpha}))$, so long as $M \geq M_1$ for
  some $M_1 = O(1/\eta^{3\gamma_2+\beta+\gamma})$.
\end{description}

Let $\ens$ be an $[N, \dcol, (1-\drow) \eta]_q$-erasure code family of
size $M$ and rate $\Rcol = \ell/N \geq 1-\dcol-O(\eta^{1/\alpha})$;
i.e., each code in $\ens$ can be used to encode an message in
$\F_q^\ell$ to a codeword in $\F_q^N$.  This, in particular, is
achieved by the result of \cref{thm:explicit} which requires
$M \geq M_0$ for some $M_0 = O(1/(\eta^{3(1+\gamma_1)}\log{q}))$ on
$M$.  The rate of the code ensemble, therefore, determines the value
of $\ell$ that $\Crow$ needs to provide.  We use the code ensemble to
correct \emph{column erasures}.  Note that once an ensemble of smaller
size $M_0$ is available, one can obtain an ensemble of the desired
size $M$ by simply repeating each code in the existing ensemble
$M/M_0$ times (assuming that $M_0 | M$).  In our case,
\cref{thm:explicit} provides an ensemble size that is a power of two
and is thus suitable for this purpose.  We consider an arbitrary
indexing of the elements of $\ens$ by the elements of $[M]$.

From a codeword of $\Crow$ in $\F_{q^\ell}^M$, we construct an
$M\times N$ matrix by interpreting the $i$th symbol of the codeword,
for $i=1,\ldots,M$, as a row vector in $\F_q^\ell$ and then encoding
the row vector to a codeword of the $i$th code
$\mathcal{C}_i\subseteq\F^N_q$ in $\ens$.  Arranging the $M$ obtained
row vectors as an $M \times N$ matrix over $\F_q$ results in the final
codeword. Note that the code over $\F_q^{M \times N}$ that we have
just described is linear over $\F_q$.

\begin{lem} \label{lem:graph:analysis} The above construction provides
  a linear $[M,N,\drow,\dcol]_q$-graph code of rate $\Rrow \Rcol$.
\end{lem}

\begin{proof}
  The claim on linearity and rate are immediate from the linearity of
  the codes $\Crow$, codes in $\ens$ and their respective rates.

  To analyze the erasure correction, consider any sets
  $S \subseteq [M]$ and $T \subseteq [N]$ where $|S| \leq \drow M$ and
  $|T| \leq \dcol N$. Let $C \in \F_q^{M \times N}$ be a codeword
  encoded as above, $c\in\F^M_{q^{\ell}}$ be the codeword in $\Crow$
  from which we get $C$, and suppose that all rows of $C$ in $S$ and
  columns in $T$ are erased. For each non-erased row $i\in[M]$ of $C$,
  we attempt to recover the corresponding symbol $c_i\in\F_{q^{\ell}}$
  using the erasure decoder of the corresponding code $\mathcal{C}_i$
  in $\ens$. By the guarantee on the fraction of codes in $\ens$ that
  succeed, this recovers all but at most an $\eta$ fraction of symbols
  $c_i$ among $i=[M]\setminus S$.  Next, we can decode the matrix to a
  vector $y\in\F_{q^\ell}^M$ where for at least $(1-\drow)(1-\eta)$
  fraction of the positions $i\in[M]$ we have $y_i=c_i$ and the rest
  of the positions $y_i$ are erased. The code $\Crow$ then ensures
  that the erased symbols can all be recovered. This completes the
  erasure correction of the codeword $C$.
\end{proof}

\paragraph{Strong Explicitness.}
In effect, in this construction, we are using the code ensemble $\ens$
over $\F_q$ to recover from column erasures and a row code achieved by
bundling $\ell/\ell_0$ copies of the single code
$\Crow\subseteq \F^M_{q^{\ell_0}}$ to recover from row erasures.  As
long as $\ell_0$ grows slowly (i.e., no more than poly-logarithmic in
the size of the matrix, which all constructions in this work satisfy),
we observe that the construction presented in this section is strongly
explicit provided that $\ens$ and $\Crow$ are both equipped with
strongly explicit constructions.

\subsection{Setting Up the Parameters} \label{sec:graph:params}

We now instantiate the construction to deduce the main result of this
section, stated below.

\begin{thm} \label{thm:matrix:explicit} Assume that outer codes are
  (strongly) $(\alpha,\beta,\gamma)$-attainable and that extractors
  are (strongly) $(\gamma_1, \gamma_2)$-attainable.  Fix any
  $(\drow, \dcol) \in [0, 1)^2$. For a parameter $\eta > 0$, there are
  $M_0 = \poly(1/\eta)$ and $N_0 = \poly(1/\eta)$ such that the
  following holds. Let $M \geq M_0$ and $N \geq N_0$ be integers where
  $M$ is a power of two. Then, there is a (strongly) explicit
  construction of a linear $[M,N,\drow,\dcol]_q$-graph code achieving
  rate at least $(1-\drow)(1-\dcol)(1-\eta^{1/\alpha})$.  Furthermore,
  after a one-time pre-processing time of $\exp(\eta^{-O(1)})$, the
  code can be encoded and decoded against erasures (as above) in
  quasi-linear time in the block length $MN$.  Concretely, when
  $\drow = 0$, one can take $M_0 = O(1/(\eta^{3(1+\gamma_1)}\log q))$
  and $N_0 = O(1/\eta^{3\gamma_2+\beta+\gamma})$; and otherwise,
  $M_0 = O(1/(\eta^{3(1+\gamma_1)}\log
  q)+1/\eta^{3\gamma_2+\beta+\gamma})$ and
  $N_0 = O(1/\eta^{3\gamma_2+\beta+\gamma}+1/\eta^{3(1+\gamma_1)})$.
\end{thm}

\begin{proof}
  We use \cref{thm:explicit} in the construction of
  \cref{sec:graph:constr} (with $\eps \coloneq (1-\drow)\eta$) for the
  code family $\ens$ and the row code $\Crow$ as defined in
  \cref{lem:graph:analysis}.  The minimum value for $N$ is given by
  the lower bound on the block length of $\ens$, and the minimum value
  for $M$ is given by the size of the code ensemble provided by
  \cref{thm:explicit}. Setting
  $N_0 = O(1/\eta^{3\gamma_2+\beta+\gamma})$ and
  $M_0 = O(1/(\eta^{3(1+\gamma_1)}\log q))$ can fulfill both
  requirements.

  We need to additionally ensure that $N_0$ is large enough to
  accommodate the minimum possible alphabet size $\ell_0$ for the code
  $\Crow$. Recall that when $\drow=0$, we have $\ell_0 = O(\log M_0)$
  and the above-mentioned choice for $N_0$ would more than suffice.
  When $\drow > 0$, we have $\ell_0 =O(1/\eta^{3(1+\gamma_1)})$. In
  order to accommodate for that we can increase the value of $N_0$
  accordingly to fulfill the requirement.

  Furthermore, we need to furthermore ensure that $M$ is large enough
  to fulfill the minimum length requirement of the code
  $\Crow$. Again, this is not an issue when $\drow = 0$ as in this
  case $\Crow$ is a Reed-Solomon code which only requires $O(1/\eta)$
  length; already accommodated by the choice of $M_0$.  When
  $\drow > 0$, recall that the minimum length requirement for $\Crow$
  is $M_1 = O(1/\eta^{3\gamma_2+\beta+\gamma})$. In this case, we
  increase our choice of $M_0$ by $M_1$ to fulfill that.

  Since the code is able to decode any codeword when there are no
  erasures, the rate of the final code is readily seen to be
  \[R = \Rrow \Rcol =
    (1-\drow)(1-O(\eta^{1/\alpha}))(1-\dcol)(1-O(\eta^{1/\alpha}))
    \geq (1-\drow)(1-\dcol)(1-O(\eta^{1/\alpha})).
  \]
  Without loss of generality, we can rewrite this as
  $(1-\drow)(1-\dcol)(1-\eta^{1/\alpha})$ as in the statement of the
  result by simply scaling $\eta$ by a constant.

  The code construction involves the pre-processing step of
  \cref{thm:explicit} followed by a straightforward implementation of
  the steps described in \cref{sec:graph:constr}, confirming (strong)
  explicitness and the encoding runtime.  Erasure decoding, as
  described in the proof of \cref{lem:graph:analysis}, also requires
  quasi-linear time since $\Crow$ and $\ens$ are equipped with
  quasi-linear time erasure decoders.
\end{proof}

Below, we record a simplified version of \cref{thm:matrix:explicit}
that additionally incorporates the result of \cref{coro:alpha} that
outer codes are strongly $(1,O(1),O(1))$-attainable, in addition to
\cref{prop:ext:constr}.

\begin{coro} \label{coro:matrix:explicit} Fix any
  $(\drow, \dcol) \in [0, 1)^2$. For a parameter $\eta > 0$, there are
  $N_0 = \poly(1/\eta)$ and $M_0 = \poly(1/\eta)$ such that the
  following holds. Let $M \geq M_0$ and $N \geq N_0$ be integers where
  $M$ is a power of two. Then, there is a strongly explicit
  construction of a linear $[M,N,\drow,\dcol]_q$-graph code achieving
  rate at least $(1-\drow)(1-\dcol)(1-\eta)$.  Furthermore, after a
  one-time pre-processing time of $\exp(\eta^{-O(1)})$, the code can
  be encoded and decoded against erasures (as above) in quasi-linear
  time in the block length $MN$.\qedhere \qed
\end{coro}

\section{Explicit Erasure Codes over Constant-Sized Alphabets}
\label{sec:AEL}

We highlight a notable special case of our construction when, in one
dimension (rows or columns), no erasures occur. Let us assume that
$\drow = 0$ so that the adversary only erases a $\dcol =: \delta$
fraction of the $N$ columns. In this case, we can interpret the
codewords (in $\F_q^{M \times N}$) of the construction provided by
\cref{thm:matrix:explicit} as elements of $\F_Q^N$, where $Q = q^M$
(i.e., each column is interpreted as an element of $\F_Q$). This
provides an $\F_q$-linear code that is nearly-MDS over a
constant-sized alphabet, akin to what constructions such as
\cite{AEL95} achieve.  Namely, we have the following result.

\begin{thm} \label{thm:almost:MDS} Assume that outer codes are
  (strongly) $(\alpha,\beta,\gamma)$-attainable and that extractors
  are (strongly) $(\gamma_1, \gamma_2)$-attainable.  Fix any
  $\delta \in [0,1)$. For a parameter $\eta > 0$, there is an
  $N_0 = O(1/\eta^{3\gamma_2+\beta+\gamma})$ and
  $Q = 2^{O(1/(\eta^{3(1+\gamma_1)})}$ (that is a power of $q$) such
  that for all $N \geq N_0$, there is a (strongly) explicit
  construction of an $\F_q$-linear code over $\F_Q$ with relative
  distance larger than $\delta$ and rate at least
  $1-\delta-\eta^{1/\alpha}$ (requiring a pre-processing time of
  $\exp(\eta^{-O(1)})$). Furthermore, the code can be encoded and
  erasure decoded (against any $\delta$ fraction of erasures) in time
  $\tilde{O}(N \log Q)$.
\end{thm}

\begin{proof}
  This is an immediate corollary of \cref{thm:matrix:explicit} with
  $\drow \coloneq 0$, $\dcol \coloneq \delta$, $M = M_0$, and
  $Q \coloneq q^M$.  Each column of each codeword in
  $\F_q^{M \times N}$ is regarded as an element of $\F_{q^M} =
  \F_Q$. resulting in an $\F_q$-linear code over $\F_Q^N$.
\end{proof}

We recall that the pre-processing time (that only depends on $1/\eta$
but exponentially so) can be eliminated using \cref{rem:exhaustive}.
Alternatively, this exhaustive search can be eliminated by using an
MDS inner code, as we explain below.

Furthermore, as an immediate consequence of our strongly explicit
construction of nearly-MDS codes over constant-sized alphabets, we can
strengthen the constructability of outer codes (i.e.,
\cref{prop:outer}), as recorded below.

\begin{coro} \label{coro:alpha} There are absolute constants
  $\beta > 1$ and $\gamma > 1$ such that outer codes are strongly
  $(1, \beta, \gamma)$-attainable. \qedhere \qed
\end{coro}
\begin{proof}
  Fix any erasure fraction $\eta_0\in [0,1)$. Using \cref{prop:outer}
  and \cref{prop:ext:constr}, we can use the parameters $\alpha=3.01$,
  $\beta=0$, $\gamma=1$, $\gamma_1=4$, $\gamma_2=2$,
  $\eta=\eta^{\alpha}_0$, and $\delta=\eta_0$ in
  \cref{thm:almost:MDS}. This obtains a strongly explicit construction
  of linear codes with rate at least $1-O(\eta_0)$, alphabet size
  $Q=2^{O(1/\eta_0^{46})}$, and minimum block length
  $N_0=O(1/\eta_0^{22})$ that can be encoded and erasure decoded
  against any $\eta_0$ fraction of erasures in time
  $\tilde{O}(N\log{Q})$. This confirms that outer codes are strongly
  $(1,46,22)$-attainable.
\end{proof}
As a consequence of this, we can use \cref{coro:alpha} and
\cref{prop:ext:constr} to instantiate all our previous results on
erasure code families, explicit codes on bipartite graphs, and in
turn, nearly-MDS codes and tighten their guarantees.  Concretely, by
picking $\alpha=1$, the size of erasure code families
\cref{thm:explicit} and the parameters of the bipartite graph code
construction in \cref{thm:matrix:explicit} can be improved.  This
layer of bootstrapping furthermore allows the choice of $\alpha = 1$
in \cref{thm:almost:MDS} as well, and is important for the application
of \cref{sec:symmetric}. Moreover, \cref{coro:alpha} simultaneously
achieves $\alpha=1$ and strong explicitness, an aspect that our
applications use.  As we have discussed before, to the best of our
knowledge, other outer codes such as expander-based constructions used
by \cite{AEL95} that also confirm a $(1,O(1),O(1))$-attainable
guarantee do not provide strong explicitness.

Even though the above result obtains a constant-sized alphabet that
only depends on the gap to the Singleton bound $\eta$, the dependence
on $\eta$ is weaker than what \cite{AEL95} obtains which is
$Q = (1/\eta)^{O(1/\eta^4)}$ for $q=2$. This is not an artifact of our
general framework. Recall that \cref{thm:almost:MDS} uses
\cref{thm:matrix:explicit}, which in turn is based on erasure code
families.  According to \cref{lem:existence:code},
$[N,\delta,\eta]_q$-erasure code families of rate $R = 1-\delta-\eta$
of size $O(1/(\eta^2 \log q))$ exist.  Such erasure code families in
our proposed constructions would result in nearly MDS codes with gap
to Singleton bound $\eta$ and alphabet size $Q = 2^{O(1/\eta^2)}$;
thus obtaining packet lengths that are quadratically better than the
\cite{AEL95} construction. In contrast, fully random (linear or
nonlinear) codes achieve the Gilbert-Varshamov bound which, in this
parameter regime (assuming $\delta>0$ is a constant), corresponds to
$Q=2^{O(1/\eta)}$. This motivates the question of improving explicit
constructions of erasure code families (or equivalently, linear seeded
extractors and lossless condensers for symbol-fixing sources).

We note that assuming optimal explicit extractors (more precisely,
extractors being $(2, O(1))$-attainable, our resulting alphabet size
is $Q=2^{O(1/\eta^9)}$.  Other than the seed length of the shuffler
extractor, we identify two sources of inefficiency for the parameters
achieved by \cref{thm:almost:MDS}:

\begin{enumerate}
\item The use of an inner erasure code family $\Ci$ (of constant block
  length $L=\poly(1/\eta)$ over $\F_q$) in the construction of
  explicit erasure code families described in \cref{sec:code:constr}.
  For the specific application of nearly-MDS codes, instead of an
  erasure inner code family $\ensIn$, we can simply use a single MDS
  code over an alphabet $\F_{q'}$ that is large enough to accommodate
  the inner code block length $L$ (namely, $q'=O(L)$ would allow the
  use of a Reed-Solomon inner code over $\F_{q'}^L$).  This eliminates
  a $1/\eta^3$ factor in the exponent of alphabet size $Q$ reported in
  \cref{thm:almost:MDS}, but induces a factor $O(\log(1/\eta))$ (to
  accommodate $\F_{q'}$ for the inner code alphabet) instead.
  Moreover, the use of an explicit inner code also eliminates the need
  for the pre-processing step needed by \cref{thm:almost:MDS} that
  constructs and tabulates the inner code ensemble.

\item Even with the inner code family replaced with one explicit inner
  code, the construction of explicit erasure code family in
  \cref{sec:code:constr} still uses two layers of averaging arguments
  (Markov's inequality) in the analysis (specifically, in the use of
  \cref{prop:balance,prop:balance:general} for the proof of
  \cref{lem:code:analysis}).  In particular, an averaging argument is
  used over the random choice of the seed of the shuffler extractor
  (that picks a code in the final erasure code family). Then, for the
  choices of the ``good'' codes in the ensemble that pass the first
  averaging argument, a second averaging argument is used on the $M$
  outer code blocks of the corresponding code in the ensemble. Once
  almost all blocks of the good codes in the ensemble are recovered,
  the outer code (applied separately to each individual code in the
  code ensemble) recovers the remaining blocks. For the specific
  application of nearly-MDS explicit codes, this is redundant, and a
  single outer code can be applied to all blocks corresponding to all
  seeds simultaneously (i.e., a total of $MD$ blocks where $D$ is the
  number of choices of the seed of the shuffler extractor).  Since the
  shuffler extractor is strong, it guarantees that all but a small
  fraction of the $MD$ blocks can be recovered, and the remaining
  blocks are recovered by a single outer code. This also eliminates
  the use of the ``row code'' $\Crow$ in the construction of
  \cref{sec:graph:constr} for this specific case (which simplifies the
  construction but does not lead to further savings in the
  asymptotics).
\end{enumerate}

Together, the above two considerations result in an improvement of
\cref{thm:almost:MDS} to what we record below.  We omit the proof
details as they involve a straightforward re-derivation of the
analysis in \cref{sec:codes} when the two considerations above are
applied.

\begin{thm} \label{thm:almost:MDS:improve} Assume that extractors are
  (strongly) $(\gamma_1, O(1))$-attainable (in particular, one can
  take $\gamma_1=4$).  Fix any $\delta \in [0,1)$. For a parameter
  $\eta > 0$, there is an $N_0 = \poly(1/\eta)$ and
  $Q = 2^{O(\log(1/\eta)/\eta^{2\gamma_1})}$ (that is a power of $q$)
  such that for all $N \geq N_0$, there is a (strongly) explicit
  construction of an $\F_q$-linear code over $\F_Q$ with relative
  distance larger than $\delta$ and rate at least $1-\delta-\eta$.
  Furthermore, the code can be encoded and erasure decoded (against
  any $\delta$ fraction of erasures) in time $\tilde{O}(N \log
  Q)$. \qedhere \qed
\end{thm}

This matches the alphabet size obtained by \cite{AEL95}; i.e.,
$Q = 2^{O(\log(1/\eta)/\eta^4)}$, assuming explicit constructions of
nearly optimal extractors (more precisely, when $\gamma_1 = 2$).

\section{Codes on Non-Bipartite Graphs}
\label{sec:symmetric}
In this section, we provide a strongly explicit construction of linear
graph codes achieving rates $R\ge(1-\sqrt{\delta})^4-o(1)$ for any
erasure ratio $\delta \in [0,1)$. We present our main theorem below.

\begin{thm}\label{thm:graph:code:explicit}
  Fix any $\delta\in[0,1)$ and a prime power $q$. For any $\eta>0$,
  there is an $N_0=(1/\eta)^{O(1)}$ such that the following holds. For
  every $N\ge N_0$, there is a strongly explicit construction of a
  linear $[N,\delta]_q$-graph code achieving rate at least
  $(1-\sqrt{\delta})^4-\eta$. Furthermore, after a one-time
  pre-processing time of $\exp{(\eta^{-O(1)})}$, the code can be
  encoded and decoded against any $\delta N$ row and column erasures
  in quasi-linear time in the block length $\binom{N}{2}$.
\end{thm}

In general, our strongly explicit construction follows the matrix
concatenation framework of \cite{KPS24}. Similarly to the framework of
\cite{KPS24}, we first choose a nearly-MDS code $\mathcal{C}_0$ with
relative distance $\sqrt{\delta}$ and rate
$1-\sqrt{\delta}-O(\eta)$. Then, we take a symmetric tensor product of
two copies of $\mathcal{C}_0$ as an outer
$[O(N), \sqrt{\delta}]_{Q}$-graph code over a large alphabet
$Q = q^{\poly(1/\eta)}$. Finally, in order to reduce the alphabet size
down to $q$, we concatenate this outer graph code with an inner code,
which is an optimal (bipartite)
$[\poly(1/\eta),\poly(1/\eta),\sqrt{\delta},\sqrt{\delta}]_q$-graph
code. The final concatenated code is an $[N,\delta]_q$-graph code with
rate at least $(1-\sqrt{\delta})^4-\eta$. There are three main
differences from \cite{KPS24}, listed below, that allow us to achieve
improved results.

\begin{enumerate}
\item While \cite{KPS24} uses a tensor product of Reed-Solomon codes
  as the outer code, we instead use our code from
  \cref{thm:almost:MDS:improve}. In order to achieve strongly explicit
  constructions, we are not able to use other nearly-MDS constructions
  such as \cite{AEL95} or algebraic geometry codes (see
  \cref{sec:overview} for a detailed discussion).

\item Our inner bipartite graph codes are the explicit codes that we
  construct in \cref{coro:matrix:explicit}, rather than those found by
  exhaustive search as in \cite{KPS24}.

\item In order to construct strongly explicit codes, \cite{KPS24}
  needs to perform concatenation three times. However, our choice of
  the outer code allows us to perform a single round of code
  concatenation, and thereby achieve improved rates.
\end{enumerate}

\begin{remark}[The Choice of the Outer Graph Code] \label{rem:kps} The
  result in \cite{KPS24} uses Reed-Solomon codes to first build
  symmetric tensor codes with zeros on the diagonal, which are then
  used as the outer graph code over large alphabets. This construction
  requires an outer code of alphabet size $O(\log^2{N})$, preventing
  the use of an exhaustive search to find a suitable inner graph
  code. Therefore, multiple layers of code concatenation are applied
  to bring down the alphabet size, resulting in a worse rate-distance
  tradeoff.  We could have started from the same outer code
  construction here as well, considering that our strongly explicit
  construction of optimal bipartite graph codes in \cref{sec:graphs}
  (namely, \cref{coro:matrix:explicit}) can accommodate any desired
  block length and leads to a single-layered code concatenation
  regardless.  However, in \cref{lem:outer:matrix:code} below we
  re-derive the argument for the more general case where merely an
  $\F_q$-linear code over a larger alphabet is available.  In this
  case, a direct tensor product over the code's actual alphabet would
  not automatically provide the required symmetry structure; instead
  requiring a tensor product over the base field $\F_q$, which makes
  the argument slightly more subtle. We choose to provide the more
  general framework that for future applications may be found
  worthwhile.
\end{remark}

\begin{proof}[Proof of \cref{thm:graph:code:explicit}]
  We set $\eps=\eta/3$ and $\delta'=\sqrt{\delta}$ to be the gap to
  capacity and the relative erasure tolerance of our inner and outer
  codewords, respectively. First, we construct a strongly explicit
  outer graph code with a relative erasure correction of $\delta'$,
  albeit over a large alphabet.

\begin{lem}[Outer Graph Code]\label{lem:outer:matrix:code}
  There are parameters $\ell'=\poly(1/\eps)$ and $n_0=\poly(1/\eps)$
  such that the following holds. For all $n\ge n_0$ and
  $\ell'\leq \ell\leq \poly(\log n)$, there is a strongly explicit
  $\F_q$-linear $[n,\delta']_{Q}$-graph code
  $\Co\subseteq \F_{Q}^{n\times n}$, where $Q=q^{\ell^2}$, with
  $\log_Q |\cC|/\binom{n+1}{2} \geq (1-\delta')^2-\eps$, implying in
  particular a rate lower bound of $\Ro\ge (1-\delta')^2-\eps$.
\end{lem}
\begin{proof}[Proof of \cref{lem:outer:matrix:code}]
  By \cref{thm:almost:MDS:improve}, there are $n_0=\poly(1/\eps)$ and
  $\ell'=\poly(1/\eps)$ such that for any $n\ge n_0$ and
  $\ell'\leq \ell\leq \poly(\log n)$, there is a strongly explicit
  $\F_q$-linear code $\mathcal{C}_0\subseteq \F_{q^\ell}^{n}$ with
  distance at least $(\delta'+\eps/4) n\ge \delta' n+2$, block length
  $n\ge n_0$, and rate $R_0\ge1-\delta'-\eps/3$. Denote by
  $A\in\F_q^{R_0\ell n\times\ell n}$ a generator matrix for
  $\mathcal{C}_0$ (as a linear code over $\F_q$).  We first consider
  an $\F_q$-linear code on symmetric matrices defined as follows
  \[
    \mathcal{C}'\coloneq\Bigl\{A^\top MA\in \F_q^{\ell n\times\ell
      n}\colon M\in \F_q^{R_0\ell n\times R_0\ell n}, \text{ where $M$
      is symmetric}\Bigl\}.
  \]
  Every codeword $C\in\mathcal{C}'$ is a symmetric matrix since
  $C^\top=(A^\top M A)^\top=A^\top M^\top A=A^\top MA=C$ for any
  symmetric $M$. Moreover, $\mathcal{C}'$ is the set of all symmetric
  matrices such that every row and column is a codeword in
  $\mathcal{C}_0$. Since $\mathrm{rank}(A)=R_0\ell n$, it follows that
  $\mathcal{C}'$ is an $\F_q$-linear code of dimension
  $\binom{1+R_0\ell n}{2}$ whose message space (i.e., the space of
  choices of $M$ in the above presentation) consists of all symmetric
  matrices having $R_0\ell n$ rows and columns. This is a symmetric
  tensor product of two $\F_q$-linear codes over the larger alphabet
  $\F_{q^\ell}$. Since the nearly-MDS code $\mathcal{C}_0$ that we use
  is only $\F_q$-linear rather than $\F_{q^{\ell}}$-linear, we have to
  define the tensor product over $\F_q$ rather than the actual
  alphabet $\F_{q^\ell}$.

  In order to obtain the graph code, we have to guarantee that each
  codeword has a zero diagonal. We can ensure this by simply
  disregarding all block-diagonal positions in the codewords of
  $\mathcal{C'}$.  Namely, since $\mathcal{C}'$ is a code over
  $\F^{\ell n\times \ell n}_q$, we can use $[n]\times[\ell]$ to index
  its rows and columns (and $([n]\times[\ell])^2$ to index its
  entries). Then, we define all positions of the form
  $\bigl((i,x),(i,y)\bigl), i\in[n],x,y\in[\ell]$ to be zeros. This
  has the effect of truncating all diagonal blocks out of the
  codewords of $\mathcal{C}'$ (and replacing them with zeros).  We
  note that this truncation does not incur any loss in rate. To see
  this, it suffices to show that the truncation cannot map a non-zero
  codeword to zero (i.e., it has a trivial kernel).  Take a non-zero
  row in any non-zero codeword of $\mathcal{C}'$, and recall that this
  is a non-zero codeword of $\mathcal{C}_0$ when interpreted as a
  vector in $(\F^{\ell}_q)^n$. Since the truncation only affects one
  of the $n$ blocks that this row contains, from the large distance of
  $\mathcal{C}_0$ we know that the row remains non-zero after the
  truncation.

  We now show that $\Co$ recovers from any $\delta' n$ ``block
  erasures'' of rows and columns. More precisely, given a non-zero
  $C \in \Co$ and any erasure sets $S,T\subseteq[n]$ of size bounded
  as $|S|,|T|\leq \delta' n$, denote
  $\overline{S} \coloneq [n] \setminus S$ and
  $\overline{T} \coloneq [n] \setminus T$. The goal is to show that
  the sub-matrix of $C$ consisting of the blocks of rows and columns
  picked by $\overline{S}$ and $\overline{T}$; in notation,
  $C_{\overline{S},\overline{T}}$, must be non-zero.  This is
  equivalent to our original assertion that $\Co$ is an $\F_q$-linear
  $[n,\delta']_{Q}$-graph code, for $Q=q^{\ell^2}$, by considering
  each codeword $C\in\F^{\ell n\times\ell n}_q$ as a symmetric matrix
  in $\F^{n\times n}_{Q}$ with a zero diagonal.

  We use $[n]\times[\ell]$ to index the rows and columns and use
  $C[(i_1,i_2),(j_1,j_2)]$ to denote the entry of $C$ indexed by
  $\bigl((i_1,i_2),(j_1,j_2)\bigl)\in([n]\times[\ell])\times([n]\times[\ell])$. Given
  a non-zero $C$, pick any non-zero row indexed by $(i_1, i_2)$.
  Recall that each row and column of $C$ are codewords of the
  $\F_q$-linear code $\cC_0$.  The erasure correction of this code
  (which is only slightly affected by the truncation; already
  accounted for) implies the existence of a non-zero entry
  $C[(i_1,i_2),(j_1,j_2)]$ where $j_1 \in \overline{T}$.  Now, we can
  use a similar argument over the non-zero column of $C$ indexed by
  $(j_1, j_2)$ to demonstrate a non-zero entry
  $C[(i'_1,i'_2),(j_1,j_2)]$, where $i'_1 \in \overline{S}$. We have
  found a non-zero entry in $C_{\overline{S},\overline{T}}$.

  To calculate the size of the code, we consider
  \[
    \frac{\log_Q |\Co|}{\binom{n+1}{2}}=\frac{\log_q
      |\cC'|}{\binom{n+1}{2} \ell^2}>\frac{\binom{1+R_0\ell
        n}{2}}{\binom{n+1}{2}\ell^2}=\frac{R_0\ell n(R_0\ell
      n+1)}{\ell^2 n(n+1)}\ge R_0^2\ge(1-\delta'-\eps/3)^2\ge
    (1-\delta')^2-\eps,\] which, in particular, results in the desired
  rate lower bound.

  Finally, since $\ell\leq \poly(\log n)$, by the strong explicitness
  of $\mathcal{C}_0$, each entry of $A$ can be computed in
  $\poly(\log n)$ time. Therefore, each entry of a generator matrix of
  $\Co$, as an $\F_q$-linear code, can also be computed in
  $\poly(\log n)$ time. This ensures that our construction is strongly
  explicit.
\end{proof}

We now concatenate the code constructed by
\cref{lem:outer:matrix:code} with a small, optimal, inner bipartite
graph code over $\F_q$ to get the final graph code $\mathcal{C}$. By
\cref{coro:matrix:explicit}, we can choose an appropriate
$\ell=\eps^{-\Theta(1)}$ and construct a linear
$[\ell/\sqrt{\Ri},\ell/\sqrt{\Ri},\delta',\delta']_q$-graph code
$\Ci\subseteq \F^{\ell/\sqrt{\Ri}\times \ell/\sqrt{\Ri}}_q$ that can
recover any $\delta'$ fraction of row and column erasures at rate
$\Ri\ge (1-\delta')^2-\eps$.  By choosing an appropriate parameter
$\ell$, we can assume that the matrix dimension
$D \coloneqq \ell/\sqrt{\Ri}$ of the inner graph code is a power of
two (as needed by \cref{coro:matrix:explicit}). The concatenation is
between the graph code $\Co$ from \cref{lem:outer:matrix:code}, as the
outer code, and $\Ci$, as the inner code.

\newcommand{\EncIn}{\mathsf{Enc}_{\mathsf{in}}}
\newcommand{\Enc}{\mathsf{Enc}} We formally describe the code
concatenation as follows. Let $\Co\subseteq\F_{Q}^{n\times n}$, where
$Q=q^{\ell^2}$, be the $\F_q$-linear $[n,\delta']_{Q}$-graph code
defined by \cref{lem:outer:matrix:code}. For any outer codeword
$C\in\Co$, we consider each entry $(i,j)\in[n]\times [n]$ of $C$ as a
matrix $C[i,j]\in\F^{\ell\times\ell}_q$. Recall that
$\Ci\subseteq \F_q^{D\times D}$ denotes a linear
$[D,D,\delta',\delta']_q$-graph code with $|\Ci|=q^{\ell^2}$. We
denote by $\EncIn\colon\F_q^{\ell\times\ell}\to \F_q^{D\times D}$ any
$\F_q$-linear encoder for $\Ci$ (defined by fixing some generator
matrix). For any outer codeword $C \in \Co$, we define
$\Enc(C)\in(\F_q)^{nD\times nD}$ as follows, using $[n]\times[D]$ to
index the rows and columns of $\Enc(C)$.
\begin{equation}\label{eq:concat}
  \Enc(C)|_{(i\times [D]),(j\times[D])}\coloneq
  \begin{cases}\EncIn(C[i,j])\quad &\text{when }i\leq j\\
    \EncIn(C[i,j]^\top)^\top \quad &\text{when $i>j$}
  \end{cases}\quad\forall (i,j)\in[n]\times[n].
\end{equation}
Observe that $\Enc(C)$ is a symmetric matrix with an all-zeros
diagonal (in fact, it has an all-zeros block diagonal).  We then
define the concatenated code $\Co\circ\Ci$ to be the graph code
$\mathcal{C} \coloneqq \{\Enc(C)\colon C\in \Co\}\subseteq
\F^{nD\times nD}_q$. Recall that $\eta=3\eps$ and $\delta=\delta'^2$,
and that $\Co$ has rate $\Ro\ge (1-\delta')^2-\eps$. Denoting
$N \coloneqq nD$, we show that $\mathcal{C}$ is the desired linear
$[N,\delta]_q$-graph code, thus completing\footnote{We note a slight
  technicality that this requires the final dimension parameter $N$ to
  be an integer multiple of the constant $D$. However, by trivially
  padding $N \times N$ matrices with additional zero rows and columns,
  $N$ can be taken to be any (large enough) integer without
  significantly affecting the rate.} the proof of
\cref{thm:graph:code:explicit}.

The rate $R$ of $\mathcal{C}$ is nearly lower bounded by the product
of the rates of the inner and outer codes, as in standard code
concatenation.  To be precise, the size of the concatenated code is
equal to the size of the outer code which, using
\cref{lem:outer:matrix:code}, leads to the rate lower bound
\begin{align}
  R &= \frac{\log_q |\Co|}{\binom{N}{2}}
      = \frac{\ell^2 \log_Q |\Co|}{\binom{N}{2}}
      = \frac{\Ri D^2 \log_Q |\Co|}{\binom{N}{2}} \label{eqn:rate:a} \\
    &\geq \frac{D^2((1-\delta')^2-\eps)^2\binom{n+1}{2}}{\binom{nD}{2}} 
      \label{eqn:rate:b} \\
    &\geq ((1-\delta')^2-\eps)^2 
      \label{eqn:rate:c} \\ 
    &\geq 
      (1-\sqrt{\delta})^4 - \eta.
      \label{eqn:rate:d}
\end{align}
Here, \cref{eqn:rate:a} uses the definition of rate, \cref{eqn:rate:b}
follows from \cref{lem:outer:matrix:code} and the bound designed for
the rate of the inner code, \cref{eqn:rate:c} follows from a simple
manipulation, and \cref{eqn:rate:d} follows from the choice of
$\delta'$ and $\eps$.  Additionally, since $\Co$ is strongly explicit
and $\Ci$ has block length $D^2=O(\ell^2)\leq\poly(\log N)$, the
concatenated code $\mathcal{C}$ is also strongly explicit.

In order to show the erasure correction guarantee, it suffices to
prove that for any non-zero codeword $C\in\mathcal{C}$ and row and
column erasure sets $E,F\subseteq [n]\times[D]$ where
$|E|,|F|\leq \delta nD$, the matrix $C|_{\overline{E},\overline{F}}$
is non-zero. Here, $\overline{E}$ and $\overline{F}$ denote
$([n]\times[D])\setminus E$ and $([n]\times[D])\setminus F$,
respectively.  Note that the definition of graph codes
(\cref{def:graph:code}) only requires recovery against matching row
and column erasure sets (i.e., when $E=F$). However, we are able to
provide a stronger guarantee of recovery from possibly distinct
erasure sets $E$ and $F$ as well.

Let $E_0\subseteq [n]$ denote the set
$\{i\in [n] \colon |\{E\cap (i\times[D])\}|> \delta' D \}$, and define
$F_0$ similarly for the column indices.  By an averaging argument, it
follows that $\max\{|E_0|,|F_0|\}<\delta' n$.  Let $C'\in\Co$ be the
outer codeword such that $\Enc(C')=C$; that is, the underlying outer
codeword from which we obtain $C$. Then, as guaranteed by
\cref{lem:outer:matrix:code}, there must exist an
$(i,j)\in([n]\setminus E_0)\times([n]\setminus F_0)$ such that
$C'[i,j]$ is non-zero. Therefore, from \cref{eq:concat}, we know that
the corresponding inner codeword $C|_{i\times[D],j\times [D]}\in\Ci$
(or its transpose) must be a non-zero codeword of $\Ci$.

Recall that $(i,j)\in ([n]\setminus E_0)\times([n]\setminus
F_0)$. Consider the non-zero submatrix
$C_{i,j} \coloneqq C|_{i\times[D],j\times[D]}$. At most $\delta' D$
rows and $\delta' D$ columns of $C_{i,j}$ are erased by the erasure
sets $E$ and $F$. Since $\Ci$ is a linear
$[D,D,\delta',\delta']_q$-graph code, it follows that $C_{i,j}$ is
non-zero even after the erasures indicated by $E$ and $F$. Thus, we
conclude that $C|_{\overline{E},\overline{F}}$ is non-zero.

Our construction requires a one-time pre-processing time of
$\exp(\eta^{-O(1)})$ for the construction of the inner bipartite graph
code (namely, \cref{coro:matrix:explicit}).  To confirm the running
time of the encoder and decoder, recall that the outer graph code is
encodable and decodable in quasi-linear time as a consequence of the
outer code underlying the tensor-based construction being encodable
and decodable in quasi-linear time. The latter is the case by
\cref{thm:almost:MDS:improve}.  Since the inner bipartite graph code
from \cref{coro:matrix:explicit} also allows quasi-linear time
encoding and erasure decoding, we conclude that our concatenated
construction can be encoded and erasure decoded in quasi-linear time
in the block length $\binom{N}{2}$.
\end{proof}

\section{Concluding Remarks}

This work studies two paradigms for achieving near-Singleton-bound
guarantees for erasure codes over constant-sized alphabets. The first
paradigm is to introduce a small amount of randomness in the code
construction (equivalently, resorting to a small \emph{family} of
codes over a fixed, such as binary, alphabet such that any erasure
pattern can be corrected by almost all codes). The second is to
increase the alphabet size to a large constant that can only depend on
the gap to capacity. While the latter paradigm has been extensively
studied, including by the celebrated work of Alon, Edmonds, and Luby
\cite{AEL95} (referred to as the AEL construction), the former has
received much less attention (e.g., \cite{Che09} is among the examples
that explicitly studies this notion). In this work, we have shown that
codes in the former paradigm imply codes in the latter
(\cref{sec:AEL}). A natural question would be to study whether the
reverse could also be true at least for specifically structured
constructions.

On a related note, we observe striking similarities between our
constructions of erasure code families \cref{sec:codes} and the AEL
construction. Both constructions essentially concatenate a
constant-sized object of the kind being constructed with an outer code
that is capable of correcting any lingering erasures (a small
fraction). While AEL deterministically rearranges the bits from
different packets (outer code symbols) into large packets using an
off-the-shelf expander graph to construct the final code, our
construction pseudo-randomly reshuffles (essentially permutes) all
bits using a randomness extractor.  Other than this broad view
similarity, the analyses for why each construction works appear
disconnected.

Nevertheless, it does appear that for the particular structure of each
construction, the underlying pseudorandom object (edge-expander graphs
of Ramanujan-type for AEL and strong extractors for ours) is
\emph{necessary} and sufficient. Both constructions (AEL and the
result of \cref{sec:AEL}) achieve a comparable alphabet size of
$\exp(\tilde{O}(1/\eta^{4}))$ when the optimal pseudorandom objects
(Ramanujan graphs for AEL and optimal strong extractors for ours) are
used.

Furthermore, and curiously, all explicit constructions known to us of
extractors for high-entropy sources that achieve seed lengths only
depending on the entropy deficiency of the source fundamentally
utilize high-quality expander graphs.  The construction \cite{RVW01}
that we have used, as well as \cite{CRVW02}, are based on the zig-zag
product constructions of expander graphs directly adapted to provide
an analogous product for extractor-type objects, but do not use
expander graphs as a black box. The closest extractor construction
resembling what AEL does is \cite{GW97}. However, this construction
additionally needs a universal family of hash functions (or a generic
strong extractor) combined with an off-the-shelf expander graphs and
achieves guarantees that are far from optimal even if optimal
(Ramanujan) expanders are used.  There are other extractor
constructions that can use off-the-shelf expander graphs, but they
utilize random walks on expander graphs, which is conceptually
different from the \emph{one-shot} bundling approach of AEL and also
do not lead to strong extractors \cite{AB09}*{Section~21.5.6}.  We
note that a formal correspondence between extractors and expander
graphs is known (\cites{Sha04,Vad10}). However, this works for
bipartite (unbalanced) vertex expanders and in a different parameter
regime (large, growing, degree) than is of interest to us and does not
appear to shed light on our inquiry.

Our work revisits the question of improved alphabet size for AEL-type
constructions. As we have shown, there are non-explicit erasure code
families that, if used in our framework, can lead to an alphabet size
$\exp(O(1/\eta^2))$ (quadratically better in the exponent than AEL),
and this motivates a continued study of erasure code families (over
binary or fixed alphabets) with improved parameters.  On the other
hand, random codes on the Gilbert-Varshamov bound achieve an alphabet
size $\exp(O(1/\eta))$ and any improvement achieving this (or beyond)
using only combinatorial tools would be considered a major
breakthrough.  Our work motivates and leaves open the question of
explicit construction of nearly-MDS codes encodable and
erasure-decodable in quasi-linear time that achieve an alphabet size
better than $\exp(\tilde{\Theta}(1/\eta^4))$ for gap to capacity
$\eta >0$.  Related to our framework, we ask for explicit strong
extractors for the high-entropy regime that extract almost all entropy
and achieve seed length $2\log(1/\eps) + f(\Delta)$ for error $\eps$
and some function $f$ of the entropy deficiency $\Delta$.

Finally, it remains an interesting open problem to construct
$[N,\delta]_q$-graph codes achieving the optimal rate
$R=(1-\delta)^2-o(1)$. We have resolved this problem for bipartite
graph codes, but the question for the non-bipartite case remains open.

\Omit{
  \section*{Acknowledgments}
  The authors thank Alexander Barg, Venkatesan Guruswami, Swastik
  Kopparty, Salil Vadhan, Chaoping Xing, Eitan Yaakobi, and David
  Zuckerman for discussions on the literature related to bounds on
  codes, codes on graphs, and extractor constructions. This research
  was partially supported by the National Science Foundation under
  Grant No.\ CCF-2236931.}

\bibliographystyle{alpha} \bibliography{references.bib}

\begin{bibdiv}
\begin{biblist}

\bib{AB09}{book}{
      author={Arora, Sanjeev},
      author={Barak, Boaz},
       title={Computational complexity: a modern approach},
   publisher={Cambridge University Press},
        date={2009},
}

\bib{AEL95}{inproceedings}{
      author={Alon, Noga},
      author={Edmonds, Jeff},
      author={Luby, Michael},
       title={Linear time erasure codes with nearly optimal recovery},
organization={IEEE},
        date={1995},
   booktitle={Proceedings of the {Annual IEEE Symposium on Foundations of
  Computer Science (FOCS)}},
       pages={512\ndash 519},
}

\bib{alo23}{article}{
      author={Alon, Noga},
      author={Gujgiczer, Anna},
      author={K{\"{o}}rner, J{\'{a}}nos},
      author={Milojevic, Aleksa},
      author={Simonyi, G{\'{a}}bor},
       title={Structured codes of graphs},
        date={2023},
     journal={{SIAM} Journal on Discrete Mathematics},
      volume={37},
      number={1},
       pages={379\ndash 403},
         url={https://doi.org/10.1137/22m1487989},
}

\bib{CDHKS00}{inproceedings}{
      author={Canetti, Ran},
      author={Dodis, Yevgeniy},
      author={Halevi, Shai},
      author={Kushilevitz, Eyal},
      author={Sahai, Amit},
       title={Exposure-resilient functions and all-or-nothing transforms},
organization={Springer},
        date={2000},
   booktitle={Proceedings of the {International Conference on the Theory and
  Application of Cryptographic Techniques (EUROCRYPT)}},
       pages={453\ndash 469},
}

\bib{CDS11}{article}{
      author={Cheraghchi, Mahdi},
      author={Didier, Fredric},
      author={Shokrollahi, Amin},
       title={Invertible extractors and wiretap protocols},
        date={2011},
     journal={IEEE Transactions on Information Theory},
      volume={58},
      number={2},
       pages={1254\ndash 1274},
}

\bib{CGHFRS85}{inproceedings}{
      author={Chor, Benny},
      author={Goldreich, Oded},
      author={H{\aa}stad, Johan},
      author={Freidmann, Joel},
      author={Rudich, Steven},
      author={Smolensky, Roman},
       title={The bit extraction problem or $t$-resilient functions},
organization={IEEE},
        date={1985},
   booktitle={Proceedings of the {Annual Symposium on Foundations of Computer
  Science (SFCS)}},
       pages={396\ndash 407},
}

\bib{CGL22}{inproceedings}{
      author={Chattopadhyay, Eshan},
      author={Goodman, Jesse},
      author={Liao, Jyun-Jie},
       title={Affine extractors for almost logarithmic entropy},
organization={IEEE},
        date={2022},
   booktitle={Proceedings of the {Annual IEEE Symposium on Foundations of
  Computer Science (FOCS)}},
       pages={622\ndash 633},
}

\bib{Che09}{inproceedings}{
      author={Cheraghchi, Mahdi},
       title={Capacity achieving codes from randomness conductors},
organization={IEEE},
        date={2009},
   booktitle={Proceedings of the {Annual IEEE International Symposium on
  Information Theory (ISIT)}},
       pages={2639\ndash 2643},
}

\bib{Che10}{thesis}{
      author={Cheraghchi, Mahdi},
       title={Applications of derandomization theory in coding},
        type={Ph.D. Thesis},
 institution={EPFL},
        date={2010},
}

\bib{CI17}{article}{
      author={Cheraghchi, Mahdi},
      author={Indyk, Piotr},
       title={Nearly optimal deterministic algorithm for sparse
  {Walsh-Hadamard} transform},
        date={2017},
     journal={ACM Transactions on Algorithms},
      volume={13},
      number={3},
       pages={1\ndash 36},
}

\bib{CRVW02}{inproceedings}{
      author={Capalbo, Michael},
      author={Reingold, Omer},
      author={Vadhan, Salil},
      author={Wigderson, Avi},
       title={Randomness conductors and constant-degree expansion beyond the
  degree/2 barrier},
        date={2002},
   booktitle={Proceedings of the {Annual ACM Symposium on Theory of Computing
  (STOC)}},
       pages={659\ndash 668},
}

\bib{CT06}{book}{
      author={Cover, T.M.},
      author={Thomas, J.A.},
       title={Elements of information theory},
     edition={Second},
   publisher={John Wiley and Sons},
        date={2006},
}

\bib{CZ19}{article}{
      author={Chattopadhyay, Eshan},
      author={Zuckerman, David},
       title={Explicit two-source extractors and resilient functions},
        date={2019},
     journal={Annals of Mathematics},
      volume={189},
      number={3},
       pages={653\ndash 705},
}

\bib{DF25}{unpublished}{
      author={Doron, Dean},
      author={Fridman, Ori},
       title={Bit-fixing extractors for almost-logarithmic entropy},
        date={2025},
        note={\emph{ECCC} Technical Report TR25-012 (available online at
  \url{https://eccc.weizmann.ac.il/report/2025/012/})},
}

\bib{Dod00}{thesis}{
      author={Dodis, Yevgeniy},
       title={Exposure-resilient cryptography},
        type={Ph.D. Thesis},
 institution={Massachusetts Institute of Technology},
        date={2000},
}

\bib{For66}{book}{
      author={Forney, G.D.},
       title={Concatenated codes},
   publisher={MIT Press},
        date={1966},
}

\bib{Fri92}{inproceedings}{
      author={Friedman, Joel},
       title={On the bit extraction problem},
organization={IEEE COMPUTER SOCIETY PRESS},
        date={1992},
   booktitle={Proceedings of the {Annual IEEE Symposium on Foundations of
  Computer Science (FOCS)}},
      volume={33},
       pages={314\ndash 314},
}

\bib{FT00}{article}{
      author={Friedl, Katalin},
      author={Tsai, {S-C.}},
       title={Two results on the bit extraction problem},
        date={2000},
     journal={Discrete applied mathematics},
      volume={99},
      number={1-3},
       pages={443\ndash 454},
}

\bib{Gab10}{book}{
      author={Gabizon, Ariel},
       title={Deterministic extraction from weak random sources},
     edition={1},
   publisher={Springer-Verlag},
     address={Berlin, Heidelberg},
        date={2010},
        ISBN={3642149022},
}

\bib{codingbook}{unpublished}{
      author={Guruswami, Venkatesan},
      author={Rudra, Atri},
      author={Sudan, Madhu},
       title={Essential coding theory},
        date={2025},
        note={Draft of the textbook available at
  \url{https://cse.buffalo.edu/faculty/atri/courses/coding-theory/book}.},
}

\bib{GS16}{article}{
      author={Guruswami, Venkatesan},
      author={Smith, Adam},
       title={Optimal rate code constructions for computationally simple
  channels},
        date={2016},
     journal={Journal of the ACM (JACM)},
      volume={63},
      number={4},
       pages={1\ndash 37},
}

\bib{GS95}{article}{
      author={Garcia, Arnaldo},
      author={Stichtenoth, Henning},
       title={A tower of {Artin-Schreier} extensions of function fields
  attaining the {Drinfeld-Vl{\u{a}}dut} bound},
        date={1995},
     journal={Inventiones Mathematicae},
      volume={121},
      number={1},
       pages={211\ndash 222},
}

\bib{Gur04}{book}{
      author={Guruswami, Venkatesan},
       title={List decoding of error-correcting codes},
   publisher={Springer Science+Business Media},
        date={2004},
      volume={3282},
}

\bib{GUV09}{article}{
      author={Guruswami, Venkatesan},
      author={Umans, Christopher},
      author={Vadhan, Salil},
       title={Unbalanced expanders and randomness extractors from
  {Parvaresh-Vardy} codes},
        date={2009},
     journal={Journal of the ACM (JACM)},
      volume={56},
      number={4},
       pages={1\ndash 34},
}

\bib{GW97}{article}{
      author={Goldreich, Oded},
      author={Wigderson, Avi},
       title={Tiny families of functions with random properties: {A}
  quality-size trade-off for hashing},
        date={1997},
     journal={Random Structures and Algorithms},
      volume={11},
      number={4},
       pages={315\ndash 343},
  url={https://doi.org/10.1002/(SICI)1098-2418(199712)11:4\<315::AID-RSA3\>3.0.CO;2-1},
}

\bib{HIV22}{inproceedings}{
      author={Huang, Xuangui},
      author={Ivanov, Peter},
      author={Viola, Emanuele},
       title={Affine extractors and {AC0-parity}},
organization={Schloss Dagstuhl--Leibniz-Zentrum f{\"u}r Informatik},
        date={2022},
   booktitle={Proceedings of the {Annual Workshop on Approximation,
  Randomization, and Combinatorial Optimization (RANDOM)}},
       pages={9\ndash 1},
}

\bib{Jus72}{article}{
      author={Justesen, J.},
       title={A class of constructive asymptotically good algebraic codes},
        date={1972},
     journal={IEEE Transactions on Information Theory},
      volume={18},
       pages={652\ndash 656},
}

\bib{KJS01}{article}{
      author={Kurosawa, Kaoru},
      author={Johansson, Thomas},
      author={Stinson, Douglas~R},
       title={Almost $k$-wise independent sample spaces and their cryptologic
  applications},
        date={2001},
     journal={Journal of Cryptology},
      volume={14},
       pages={231\ndash 253},
}

\bib{KPS24}{inproceedings}{
      author={Kopparty, Swastik},
      author={Potukuchi, Aditya},
      author={Sha, Harry},
       title={Error-correcting graph codes},
        date={2025},
   booktitle={Proceedings of the annual {Conference on Innovations in
  Theoretical Computer Science (ITCS)}},
      editor={Meka, Raghu},
      series={LIPIcs},
      volume={325},
   publisher={Schloss Dagstuhl - Leibniz-Zentrum f{\"{u}}r Informatik},
       pages={67:1\ndash 67:20},
         url={https://doi.org/10.4230/LIPIcs.ITCS.2025.67},
}

\bib{LCGSW19}{inproceedings}{
      author={Lin, Fuchun},
      author={Cheraghchi, Mahdi},
      author={Guruswami, Venkatesan},
      author={Safavi-Naini, Reihaneh},
      author={Wang, Huaxiong},
       title={Secret sharing with binary shares},
        date={2019},
   booktitle={Proceedings of the annual {Conference on Innovations in
  Theoretical Computer Science (ITCS)}},
      editor={Blum, Avrim},
      series={Leibniz International Proceedings in Informatics (LIPIcs)},
      volume={124},
   publisher={Schloss Dagstuhl -- Leibniz-Zentrum f{\"u}r Informatik},
     address={Dagstuhl, Germany},
       pages={53:1\ndash 53:20},
  url={https://drops.dagstuhl.de/entities/document/10.4230/LIPIcs.ITCS.2019.53},
}

\bib{Li16}{inproceedings}{
      author={Li, Xin},
       title={Improved two-source extractors, and affine extractors for
  polylogarithmic entropy},
organization={IEEE},
        date={2016},
   booktitle={Proceedings of the {Annual IEEE Symposium on Foundations of
  Computer Science (FOCS)}},
       pages={168\ndash 177},
}

\bib{LLLEWX24}{article}{
      author={Li, Songsong},
      author={Liu, Shu},
      author={Ma, Liming},
      author={Wan, Yunqi},
      author={Xing, Chaoping},
       title={Encoding of algebraic geometry codes with quasi-linear complexity
  {$O(N \log N)$}},
        date={2024},
     journal={Preprint arXiv:2407.04618},
}

\bib{Mas63}{techreport}{
      author={Massey, James~L.},
       title={Threshold decoding},
 institution={Massachusetts Institute of Technology, Research Laboratory of
  Electronics},
        date={1963},
}

\bib{MS77}{book}{
      author={MacWilliams, F.J.},
      author={Sloane, N.J.},
       title={The theory of error-correcting codes},
   publisher={North Holand},
        date={1977},
}

\bib{NW19}{article}{
      author={Narayanan, Anand~Kumar},
      author={Weidner, Matthew},
       title={Subquadratic time encodable codes beating the {Gilbert-Varshamov}
  bound},
        date={2019},
     journal={IEEE Transactions on Information Theory},
      volume={65},
      number={10},
       pages={6010\ndash 6021},
}

\bib{NZ96}{article}{
      author={Nisan, Noam},
      author={Zuckerman, David},
       title={Randomness is linear in space},
        date={1996},
     journal={Journal of Computer and System Sciences},
      volume={52},
      number={1},
       pages={43\ndash 52},
}

\bib{PR11}{article}{
      author={Porat, Ely},
      author={Rothschild, Amir},
       title={Explicit nonadaptive combinatorial group testing schemes},
        date={2011},
     journal={IEEE Transactions on Information Theory},
      volume={57},
      number={12},
       pages={7982\ndash 7989},
}

\bib{Rao09}{inproceedings}{
      author={Rao, Anup},
       title={Extractors for low-weight affine sources},
organization={IEEE},
        date={2009},
   booktitle={Proceedings of the 24th {Annual IEEE Conference on Computational
  Complexity (CCC)}},
       pages={95\ndash 101},
}

\bib{Riv97}{inproceedings}{
      author={Rivest, Ronald~L},
       title={All-or-nothing encryption and the package transform},
organization={Springer},
        date={1997},
   booktitle={Proceedings of the {International Workshop on Fast Software
  Encryption (FSE)}},
       pages={210\ndash 218},
}

\bib{criss97}{article}{
      author={Roth, Ron~M.},
       title={Probabilistic crisscross error correction},
        date={1997},
     journal={{IEEE} Transactions on Information Theory},
      volume={43},
      number={5},
       pages={1425\ndash 1438},
         url={https://doi.org/10.1109/18.623142},
}

\bib{RRV99}{inproceedings}{
      author={Raz, Ran},
      author={Reingold, Omer},
      author={Vadhan, Salil},
       title={Extracting all the randomness and reducing the error in
  {Trevisan's} extractors},
        date={1999},
   booktitle={Proceedings of the {Annual ACM Symposium on Theory of Computing
  (STOC)}},
       pages={149\ndash 158},
}

\bib{RTS00}{article}{
      author={Radhakrishnan, Jaikumar},
      author={Ta-Shma, Amnon},
       title={Bounds for dispersers, extractors, and depth-two
  superconcentrators},
        date={2000},
     journal={SIAM Journal on Discrete Mathematics},
      volume={13},
      number={1},
       pages={2\ndash 24},
}

\bib{RVW00}{inproceedings}{
      author={Reingold, Omer},
      author={Vadhan, Salil},
      author={Wigderson, Avi},
       title={Entropy waves, the zig-zag graph product, and new constant-degree
  expanders and extractors},
organization={IEEE},
        date={2000},
   booktitle={Proceedings {Annual IEEE Symposium on Foundations of Computer
  Science (FOCS)}},
       pages={3\ndash 13},
}

\bib{RVW01}{unpublished}{
      author={Reingold, Omer},
      author={Vadhan, Salil},
      author={Wigderson, Avi},
       title={Entropy waves, the zig-zag graph product, and new constant-degree
  expanders and extractors},
        date={2001},
        note={\emph{ECCC} Technical Report TR01-018 (available online at
  \url{https://eccc.weizmann.ac.il/report/2001/018/})},
}

\bib{SAKS01}{article}{
      author={Shum, Kenneth~W.},
      author={Aleshnikov, Ilia},
      author={Kumar, P.~Vijay},
      author={Stichtenoth, Henning},
      author={Deolalikar, Vinay},
       title={A low-complexity algorithm for the construction of
  algebraic-geometric codes better than the {Gilbert-Varshamov} bound},
        date={2001},
     journal={IEEE Transactions on Information Theory},
      volume={47},
      number={6},
       pages={2225\ndash 2241},
}

\bib{Sha04}{article}{
      author={Shaltiel, Ronen},
       title={Recent developments in explicit constructions of extractors},
        date={2004},
     journal={Current Trends in Theoretical Computer Science},
       pages={189\ndash 228},
}

\bib{Spi95}{inproceedings}{
      author={Spielman, Daniel~A.},
       title={Linear-time encodable and decodable error-correcting codes},
        date={1995},
   booktitle={Proceedings of the {Annual ACM Symposium on Theory of Computing
  (STOC)}},
       pages={388\ndash 397},
}

\bib{SS96}{article}{
      author={Sipser, Michael},
      author={Spielman, Daniel~A.},
       title={Expander codes},
        date={1996},
     journal={IEEE Transactions on Information Theory},
      volume={42},
      number={6},
       pages={1710\ndash 1722},
}

\bib{Sti93}{article}{
      author={Stinson, Douglas~R.},
       title={Resilient functions and large sets of orthogonal arrays},
        date={1993},
     journal={Congressus Numerantium},
       pages={105\ndash 105},
}

\bib{TVZ82}{article}{
      author={Tsfasman, Michael~A},
      author={Vl{\u{a}}dut, SG},
      author={Zink, Th},
       title={Modular curves, {Shimura} curves, and {Goppa} codes, better than
  {Varshamov-Gilbert} bound},
        date={1982},
     journal={Mathematische Nachrichten},
      volume={109},
      number={1},
       pages={21\ndash 28},
}

\bib{Vad10}{inproceedings}{
      author={Vadhan, S.},
       title={The unified theory of pseudorandomness},
        date={2010},
   booktitle={Proceedings of the {International Congress of Mathematicians
  (ICM)}},
}

\bib{yey20}{article}{
      author={Yohananov, Lev},
      author={Efron, Yuval},
      author={Yaakobi, Eitan},
       title={Double and triple node-erasure-correcting codes over complete
  graphs},
        date={2020},
     journal={IEEE Transactions on Information Theory},
      volume={66},
      number={7},
       pages={4089\ndash 4103},
}

\bib{yy19}{article}{
      author={Yohananov, Lev},
      author={Yaakobi, Eitan},
       title={Codes for graph erasures},
        date={2019},
     journal={{IEEE} Transactions on Information Theory},
      volume={65},
      number={9},
       pages={5433\ndash 5453},
         url={https://doi.org/10.1109/TIT.2019.2910040},
}

\end{biblist}
\end{bibdiv}

\end{document}